\documentclass[11pt,draftcls,onecolumn]{IEEEtran}
\usepackage{graphicx}
\usepackage{amsmath}
\usepackage{amsthm}
\usepackage{mathrsfs}
\usepackage{mathtools}
\usepackage{amssymb}
\usepackage{float}
\usepackage{url}
\usepackage{verbatim}
\usepackage{dsfont}
\usepackage{enumerate}
\usepackage{layout}

\usepackage[bottom=0.8in]{geometry}

\newtheorem{theorem}{Theorem}
\newtheorem{lemma}{Lemma}
\newtheorem{proposition}{Proposition}
\newtheorem{corollary}{Corollary}

\newtheorem{definition}{Definition}

\newcommand{\prob}{\ensuremath{\mathbb{P}}}
\newcommand{\naturals}{\ensuremath{\mathbb{N}}}
\newcommand{\Reals}{\ensuremath{\mathbb{R}}}
\newcommand{\expectation}{\ensuremath{\mathbb{E}}}
\newcommand{\supp}{\mathop{\mathrm{supp}}}
\newcommand{\set}{\ensuremath{\mathcal}}
\newcommand{\tcl}[1]{\set{T}^{(#1)}}
\newcommand{\typ}[1]{\set{T}_{#1}}
\newcommand\Rentr[1]{H_{#1}}

\newcommand{\OneTo}[1]{[1:#1]}
\newcommand{\FromTo}[2]{[#1:#2]}  

\DeclareMathOperator{\identity}{id}

\DeclareMathOperator{\IID}{i.i.d.}

\newcommand{\tmax}{\textnormal{max}}

\allowdisplaybreaks

\begin{document}
\thispagestyle{empty}
\setcounter{page}{1}
\setlength{\baselineskip}{1.15\baselineskip}

\title{\huge{On Two-Stage Guessing}}
\author{Robert Graczyk \qquad Igal Sason \\
\thanks{R. Graczyk is with the Signal and Information Processing Laboratory,
ETH Zurich, 8092 Zurich, Switzerland (e-mail: grackzyk@isi.ee.ethz.ch).}
\thanks{I. Sason is with the Andrew and Erna Viterbi Faculty of Electrical and Computer Engineering, Technion--Israel
Institute of Technology, Haifa 3200003, Israel (e-mail: sason@ee.technion.ac.il).}
\thanks{The material in Section~\ref{section: Robert's part} was presented in part at
the {\em 2019~IEEE International Symposium on Information Theory}, Paris, France, July 2019.}
\thanks{Citation for this work: \newline
R. Graczyk and I. Sason, ``On two-stage guessing,'' {\em Information}, vol.~12, no.~4,
paper~159, pp.~1--20, April 2021.}}

\maketitle
\thispagestyle{empty}

\vspace*{-0.8cm}
\begin{abstract}
Stationary memoryless sources produce two correlated random sequences $X^n$ and $Y^n$. A guesser seeks
to recover $X^n$ in two stages, by first guessing $Y^n$ and then $X^n$. The contributions of this work are twofold:
(1)~We characterize the least achievable exponential growth rate (in $n$) of any positive $\rho$-th moment
of the total number of guesses when $Y^n$ is obtained by applying a deterministic function $f$ component-wise
to $X^n$. We prove that, depending on $f$, the least exponential growth rate in the two-stage setup is lower
than when guessing $X^n$ directly. We further propose a simple Huffman code-based construction of a function
$f$ that is a viable candidate for the minimization of the least exponential growth rate in the two-stage guessing setup.
(2)~We characterize the least achievable exponential growth rate of the $\rho$-th moment of the total number of
guesses required to recover $X^n$ when Stage~1 need not end with a correct guess of $Y^n$ and without assumptions
on the stationary memoryless sources producing $X^n$ and~$Y^n$.
\end{abstract}


\section{Introduction}
\label{section: Introduction}

Pioneered by Massey \cite{Massey94}, McEliece and Yu \cite{McElieceYu95}, and Arikan \cite{Arikan96},
the guessing problem is concerned with recovering the realization of a finite-valued random variable $X$
using a sequence of yes-no questions of the form ``Is $X=x_1$?'', ``Is $X=x_2$?'', etc., until correct.
A~commonly used performance metric for this problem is the $\rho$-th moment of the number of guesses until
$X$ is revealed (where $\rho$ is a positive parameter).

When guessing a length-$n$ $\IID$ sequence $X^n$ (a tuple of $n$ components that are drawn independently
according to the law of $X$), the $\rho$-th moment of the number of guesses required to recover the realization of
$X^n$ grows exponentially with $n$, and the exponential growth rate is referred to as the guessing exponent.
The least achievable guessing exponent was derived by Arikan \cite{Arikan96}, and it equals the order-$\frac1{1+\rho}$
R\'{e}nyi entropy of $X$. Arikan's result is based on the optimal deterministic guessing strategy, which proceeds
in descending order of the probability mass function (PMF) of $X^n$.

In this paper, we propose and analyze a two-stage guessing strategy to recover the realization of an $\IID$ sequence $X^n$.
In Stage~1, the guesser is allowed to produce guesses of an ancillary sequence $Y^n$ that is jointly $\IID$ with $X^n$.
In Stage~2, the guesser must recover $X^n$. We show the following:
\begin{enumerate}[1)]
 \item When $Y^n$ is generated by component-wise application
of a mapping $f \colon \set{X} \to \set{Y}$ to $X^n$ and the guesser is required to recover $Y^n$ in Stage~1 before proceeding
to Stage~2, then the least achievable guessing exponent (i.e., the exponential growth rate of the $\rho$-th moment of the total
number of guesses in the two stages) equals
\begin{equation}\label{eq:first}
 \rho \, \max \Bigl\{ \, H_{\frac1{1+\rho}}\bigl(f(X)\bigr), \;
 H_{\frac1{1+\rho}}\bigl(X \, | \, f(X) \bigr) \Bigr\},
\end{equation}
where the maximum is between the order-$\frac{1}{1 + \rho}$ R\'{e}nyi entropy of $f(X)$ and the conditional Arimoto--R\'{e}nyi entropy of $X$ given $f(X)$. We derive \eqref{eq:first} in Section \ref{section: guess Y^n first} and summarize our analysis in Theorem~\ref{theorem1: 2-stage guessing}. We also propose a Huffman code-based
construction of a function $f$ that is a viable candidate for the minimization of \eqref{eq:first} among all maps from $\set{X}$ to $\set{Y}$ (see Algorithm 2 and Theorem \ref{theorem2: 2-stage guessing}.)
 \item When $X^n$ and $Y^n$ are jointly i.i.d. according to the PMF $P_{XY}$ and Stage~1 need not end with a correct guess of $Y^n$ (i.e., the guesser may proceed to Stage 2 even if $Y^n$ remains unknown), then the least achievable guessing exponent equals
\begin{equation}\label{eq:second}
 \sup_{Q_{XY}} \Big(\rho\min\big\{H(Q_X), \max\big\{H(Q_Y), H(Q_{X \mid Y})\big\}\big\} - D(Q_{XY} || P_{XY})\Big),
\end{equation}
where the supremum is over all PMFs $Q_{XY}$ defined on the same set as $P_{XY}$; and $H(\cdot)$ and $D(\cdot \| \cdot)$ denote, respectively, the (conditional) Shannon entropy and the Kullback--Leibler divergence. We derive \eqref{eq:second} in Section \ref{section: Robert's part} and summarize our analysis in \mbox{Theorem \ref{thm:main_robert}}. Parts of Section \ref{section: Robert's part} were presented in the conference paper~\cite{GrackzykL19}.
\end{enumerate}

Our interest in the two-stage guessing problem is due to its relation to information measures: Analogous to how the R\'{e}nyi entropy can be defined operationally via guesswork, as opposed to its axiomatic definition, we view \eqref{eq:first} and \eqref{eq:second} as quantities that capture at what cost and to what extent knowledge of $Y$ helps in recovering $X$. For example, minimizing \eqref{eq:first} over descriptions $f(X)$ of $X$ can be seen as isolating the most beneficial information of $X$ in the sense that describing it in any more detail is too costly (the first term of the maximization in \eqref{eq:first} exceeds the second), whereas a coarser description leaves too much uncertainty (the second term exceeds the first). Similarly, but with the joint law of $(X, Y)$ fixed, \eqref{eq:second} quantifies the least (partial) information of $Y$ that benefits recovering $X$ (because an optimal guessing strategy will proceed to Stage 2 when guessing $Y$ no longer benefits guessing $X$). Note that while \eqref{eq:first} and \eqref{eq:second} are derived in this paper, studying their information-like properties is a subject of future research (see Section~\ref{sec:conclusion}.)

Besides its theoretic implications, the guessing problem is also applied practically in communications and cryptography. This includes sequential decoding
(\cite{ArikanM98-2, Boztas97}), and measuring password strength \cite{Cachin97}, confidentiality
of communication channels \cite{ArikanM99}, and resilience against brute-force attacks~\cite{BracherHL_IT19}.
It is also strongly related to task encoding and (lossless and lossy) compression (see, e.g.,
\cite{ArikanM98-1, BracherLP_IT17, BracherLP_Entropy19, ChristiansenD13, Sason18b, SasonV18b}).

Variations of the guessing problem include guessing under source uncertainty \cite{Sundaresan07}, distributed guessing \cite{MC20, Salam17}, and guessing
on the Gray--Wyner and the Slepian--Wolf network \cite{GrackzykL20}.

\section{Preliminaries}
\label{section: Preliminaries}

We begin with some notation and preliminary material that are essential
for the presentation in Section~\ref{section: guess Y^n first} ahead.
The analysis in Section~\ref{section: Robert's part} relies on the method
of types (see, e.g., Chapter~11 in \cite{Cover_Thomas}).

Throughout the paper, we use the following notation:
\begin{itemize}
\item For $m, n \in \naturals$ with $m<n$, let $\FromTo{m}{n} := \{m, \ldots, n\}$;
\item Let $P$ be a PMF that is defined on a finite set $\set{X}$. For
$k \in \OneTo{|\set{X}|}$, let $G_P(k)$ denote the sum of its $k$ largest point masses, and let
$p_\tmax := G_P(1)$. For $n \in \naturals$, denote by $\set{P}_n$ the set of all PMFs defined on $\OneTo{n}$.
\end{itemize}

The next definitions and properties are related to majorization and R\'{e}nyi measures.
\begin{definition}[{Majorization}]
\label{definition: majorization}
{Consider PMFs $P$ and $Q$, defined on the same
(finite or countably infinite) set $\set{X}$. We say that $Q$ {majorizes} $P$,
denoted $P \prec Q$, if $G_P(k) \leq G_Q(k)$ for all $k \in \OneTo{|\set{X}|}$.
If $P$ and $Q$ are defined on finite sets of different cardinalities, then the PMF
defined on the smaller set is zero-padded to match the cardinality of the larger~set.}
\end{definition}

By Definition~\ref{definition: majorization}, a unit mass majorizes
any other distribution; on the other hand, the uniform distribution (on a finite
set) is majorized by any other distribution of equal~support.

\begin{definition}[{Schur-convexity/concavity}]
{A function $f \colon \set{P}_n \to \Reals$ is {Schur-convex} if, for
every $P,Q \in \set{P}_n$ with $P \prec Q$, we have $f(P) \leq f(Q)$. Likewise,
$f$ is {Schur-concave} if $-\!f$ is Schur-convex, i.e., if $P \prec Q$ implies
that $f(P) \geq f(Q)$.}
\end{definition}

\begin{definition}[{R\'{e}nyi entropy} \cite{Renyientropy}]
\label{definition: Renyi entropy}
{Let $X$ be a random variable taking values on a finite or countably infinite set
$\set{X}$ according to the PMF $P_X$. The order-$\alpha$ R\'{e}nyi entropy
$H_{\alpha}(X)$ of $X$ is given by
\begin{equation} \label{eq: Renyi entropy}
 H_{\alpha}(X) := \frac1{1-\alpha} \, \log \left( \, \sum_{x \in \set{X}}
 P_X^{\alpha}(x)\right),\quad \alpha \in (0, 1) \cup (1, \infty),
\end{equation}
where unless explicitly given, the base of $\log(\cdot)$ can be chosen arbitrarily,
with $\exp(\cdot)$ denoting its inverse function. Via continuous extension,
\begin{align}
\label{eq: RE of zero order}
& H_0(X) := \log \, \bigl| \{x \in \set{X} \colon
P_X(x) > 0 \} \bigr|, \\
\label{eq: Shannon entropy}
& H_1(X) := H(X) = - \sum_{x \in \set{X}} P_X(x) \, \log P_X(x), \\
\label{eq: RE of infinite order}
& H_{\infty}(X) := \log \frac1{p_\tmax},
\end{align}
where $H(X)$ is the (Shannon) entropy of $X$.}
\end{definition}

\begin{proposition}[Schur-concavity of the R\'{e}nyi entropy, Appendix F.3.a of \cite{MarshallOA}]
\label{proposition: Renyi entropy is Schur-concave}
The R\'{e}nyi entropy of any order $\alpha > 0$ is Schur-concave
(in particular, the Shannon entropy is Schur-concave).
\end{proposition}

\begin{definition}[{Arimoto-R\'{e}nyi conditional entropy} \cite{Arimoto75}]
\label{definition: AR conditional entropy}
{Let $(X, Y)$ be a pair of random variables taking values on a product
set~$\set{X} \times \set{Y}$ according to the PMF $P_{XY}$. When $\set{X}$ is
finite or countably infinite, the order-$\alpha$ {Arimoto--R\'{e}nyi
conditional entropy} $H_{\alpha}(X | Y)$ of $X$ given $Y$ is defined as follows:
\begin{itemize}
\item
If $\alpha \in (0,1) \cup (1, \infty) $,
\begin{align}
\label{eq1: Arimoto - cond. RE}
H_{\alpha}(X | Y) &:= \frac{\alpha}{1-\alpha} \,
\log \, \mathbb{E} \left[
\left( \, \sum_{x \in \set{X}} P_{X|Y}^{\alpha}(x|Y)
\right)^{\frac1{\alpha}} \right].
\end{align}
When $\set{Y}$ is finite, \eqref{eq1: Arimoto - cond. RE} can be
simplified as follows:
\begin{align}
H_{\alpha}(X | Y) &= \frac{\alpha}{1-\alpha} \, \log
\sum_{y \in \set{Y}} \left( \, \sum_{x \in \set{X}} P_{XY}^{\alpha}(x,y)
\right)^{\frac1{\alpha}} \label{eq2a: Arimoto - cond. RE} \\
\label{eq2b: Arimoto - cond. RE}
&= \frac{\alpha}{1-\alpha} \, \log
\, \sum_{y \in \set{Y}} P_Y(y) \, \exp \left(
\frac{1-\alpha}{\alpha} \;
H_{\alpha}(X | Y=y) \right).
\end{align}
\item If $\alpha \in \{0, 1, \infty\}$ and $\set{Y}$ is finite, then, via continuous extension,
\begin{align}
\label{eq: cond. RE at 0}
H_0(X|Y)
&= \log \, \max_{y \in \set{Y}} \, \bigl| \supp P_{X|Y}(\cdot|y) \bigr| = \max_{y \in \set{Y}} \,
H_0(X \, | \, Y=y), \\[0.1cm]
\label{eq: cond. RE at 1}
H_1(X|Y) &= H(X|Y), \\[0.1cm]
\label{eq: cond. RE at infinity}
H_{\infty}(X|Y) &= \log \, \frac1{\sum_{y \in \set{Y}}
\Bigl\{ \underset{x \in \set{X}}{\max} \, P_{X|Y}(x|Y) \cdot P_Y(y) \Bigr\}}.
\end{align}
\end{itemize}}
\end{definition}
The properties of the Arimoto--R\'{e}nyi conditional entropy were studied in \cite{FehrB14,SasonV18a}.

Finally, $\set{F}_{n,m}$ with $n,m \in \naturals$ denotes the set of all {deterministic}
functions $f \colon \OneTo{n} \to \OneTo{m}$. If $m<n$, then a
function $f \in \set{F}_{n,m}$ is not one-to-one (i.e, it is a non-injective function).

\section{Two-Stage Guessing: $Y_i = f(X_i)$}
\label{section: guess Y^n first}

Let $X^n := (X_1, \ldots, X_n)$ be a sequence of i.i.d. random variables taking
values on a finite set~$\set{X}$. Assume without loss of generality that
$\set{X} = \OneTo{|\set{X}|}$. Let $m \in \FromTo{2}{|\set{X}|-1}$,
$\set{Y} := \OneTo{m}$, and let $f \colon \set{X} \to \set{Y}$ be a fixed
deterministic function (i.e., $f \in \set{F}_{|\set{X}|, m}$).
Consider guessing $X^n \in \set{X}^n$ in two stages as follows.

\newpage
{\underline{Algorithm 1 (two-stage guessing algorithm)}}:
\begin{enumerate}[a)]
\item {Stage~1}: $Y^n := \bigl( f(X_1), \ldots, f(X_n) \bigr) \in \set{Y}^n$
is guessed by asking questions of the form:
\begin{itemize}
 \item[] ``Is $Y^n = \widehat{\boldsymbol{y}}_1$?'',
 ``Is $Y^n = \widehat{\boldsymbol{y}}_2$?'', \ldots
\end{itemize}
until correct. Note that as $|\set{Y}| = m < |\set{X}|$, this stage cannot reveal $X^n$.
\item {Stage~2}: Based on $Y^n$, the sequence $X^n \in \set{X}^n$ is guessed
by asking questions of the~form:
\begin{itemize}
 \item[] ``Is $X^n = \widehat{\boldsymbol{x}}_1$?'',
 ``Is $X^n = \widehat{\boldsymbol{x}}_2$?'', \ldots
\end{itemize}
until correct. If $Y^n = y^n$, the guesses
$\widehat{\boldsymbol{x}}_k := (\widehat{x}_{k, 1}, \ldots,
\widehat{x}_{k, n})$ are restricted to $\set{X}$-sequences, which
satisfy $f(\widehat{x}_{k, i}) = y_i$ for all $i \in \OneTo{n}$.
\end{enumerate}

The guesses $\widehat{\boldsymbol{y}}_1$, $\widehat{\boldsymbol{y}}_2$, \ldots,
in Stage 1 are in descending order of probability as measured by $P_{Y^n}$
(i.e., $\widehat{\boldsymbol{y}}_1$ is the most probable
sequence under $P_{Y^n}$; $\widehat{\boldsymbol{y}}_2$ is the second
most probable; and so on; ties are resolved arbitrarily). We denote
the index of $y^n \in \set{Y}^n$ in this guessing order by $g_{Y^n}(y^n)$. Note that
because every sequence $y^n$ is guessed exactly once, $g_{Y^n}(\cdot)$ is a
bijection from $\set{Y}^n$ to $\OneTo{m^n}$; we refer to such bijections as
{ranking functions}. The guesses $\widehat{\boldsymbol{x}}_1$,
$\widehat{\boldsymbol{x}}_2$, \ldots, in Stage 2 depend on $Y^n$ and are in
descending order of the posterior $P_{X^n | Y^n}(\cdot | Y^n)$. Following our
notation from Stage 1, the index of $x^n \in \set{X}^n$ in the guessing order
induced by $Y^n = y^n$ is denoted $g_{X^n | Y^n}(x^n | y^n)$. Note that for
every $y^n \in \set{Y}^n$, the function $g_{X^n | Y^n}(\cdot | y^n)$ is a ranking
function on $\set{X}^n$. Using $g_{Y^n}(\cdot)$ and $g_{X^n | Y^n}(\cdot |
\cdot)$, the total number of guesses $G_2(X^n)$ in Algorithm~1 can be
expressed as
\begin{align} \label{G1}
 G_2(X^n) = g_{Y^n}(Y^n) + g_{X^n | Y^n}(X^n | Y^n),
\end{align}
where $g_{Y^n}(Y^n)$ and $g_{X^n | Y^n}(X^n | Y^n)$ are the number of guesses in
Stages 1 and 2, respectively. Observe that guessing in descending order of probability
minimizes the $\rho$-th moment of the number of guesses in both stages of Algorithm 1.
By \cite{Arikan96}, for every $\rho > 0$, the guessing moments $\expectation\bigl[g_{Y^n}^{\rho}(Y^n)\bigr]$
and $\expectation\bigl[g_{X^n|Y^n}^{\rho}(X^n|Y^n)\bigr]$ can be (upper and lower) bounded
in terms of $H_{\frac1{1+\rho}}(Y)$ and $H_{\frac1{1+\rho}}(X | Y)$ as follows:
\begin{subequations}\label{eq:arikan_bounds}
\begin{align}
\bigl( 1 + n \ln m \bigr)^{-\rho} \; \exp\Bigl( n \rho \, H_{\frac1{1+\rho}}(Y)
\Bigr) &\leq \expectation\bigl[g_{Y^n}^{\rho}(Y^n)\bigr]\nonumber\\
&\leq \exp\Bigl( n \rho \,
H_{\frac1{1+\rho}}(Y) \Bigr),\label{Arikan96-LB1-UB1}\\[1em]
\bigl(1 + n \ln |\set{X}| \bigr)^{-\rho} \, \exp\Bigl( n \rho \,
H_{\frac1{1+\rho}}\bigl(X \, | \, f(X) \bigr) \Bigr)
&\leq \expectation\bigl[g_{X^n|Y^n}^{\rho}(X^n|Y^n)\bigr]\nonumber\\
&\leq \exp\Bigl( n \rho \, H_{\frac1{1+\rho}}\bigl(X \, | \, f(X) \bigr) \Bigr).
\label{Arikan96-LB2-UB2}
\end{align}
\end{subequations}

Combining \eqref{G1} and \eqref{eq:arikan_bounds}, we next establish bounds on
$\expectation[G_2(X^n)^\rho]$. In light of \eqref{G1}, we begin with bounds on
the $\rho$-th power of a sum.
\begin{lemma} \label{lemma1}
Let $k \in \naturals$, and let $\{a_i\}_{i=1}^k$ be a non-negative sequence.
For every $\rho > 0$,
\begin{align} \label{lemma1: 1}
s_1(k,\rho) \sum_{i=1}^k a_i^\rho \leq \left( \sum_{i=1}^k a_i \right)^\rho
\leq s_2(k,\rho) \sum_{i=1}^k a_i^\rho,
\end{align}
where
\begin{align} \label{lemma1: s1}
s_1(k,\rho) :=
\begin{cases}
1, & \textnormal{if} \; \rho \geq 1 \\
k^{\rho-1}, & \textnormal{if} \; \rho \in (0,1),
\end{cases}
\end{align}
and
\begin{align} \label{lemma1: s2}
s_2(k,\rho) :=
\begin{cases}
k^{\rho-1}, & \textnormal{if} \; \rho \geq 1 \\
1, & \textnormal{if} \; \rho \in (0,1).
\end{cases}
\end{align}
\begin{itemize}
\item
If $\rho \geq 1$, then the left and right inequalities in \eqref{lemma1: 1} hold
with equality if, respectively, $k-1$ of the $a_i$'s are equal to zero or $a_1 =
\ldots = a_k$;
\item
if $\rho \in (0,1)$, then the left and right inequalities in \eqref{lemma1: 1}
hold with equality if, respectively, $a_1 = \ldots = a_k$ or $k-1$ of the
$a_i$'s are equal to zero.
\end{itemize}
\end{lemma}

\begin{proof}
See Appendix~\ref{appendix: proof of Lemma 1}.
\end{proof}

\bigskip
Using the shorthand notation $k_1(\rho) := s_1(2, \rho)$, $k_2(\rho) := s_2(2,
\rho)$, we apply \mbox{Lemma \ref{lemma1}} in conjunction with \eqref{G1} and
\eqref{eq:arikan_bounds} (and the fact that $|\set{X}| \geq m$) to bound
$\expectation \bigl[ G_2^{\rho}(X^n) \bigr]$ as~follows:
\begin{subequations}\label{eq:raw_bounds_on_moment}
\begin{align}
& k_1(\rho)\bigl( 1 + n \ln |\set{X}| \bigr)^{-\rho} \left[
\exp\Bigl( n \rho \, H_{\frac1{1+\rho}}\bigl(f(X)\bigr) \Bigr)
+ \exp\Bigl( n \rho \,
H_{\frac1{1+\rho}}\bigl(X \, | \, f(X) \bigr) \Bigr) \right] \nonumber \\
& \quad \leq \expectation \bigl[ G_2^{\rho}(X^n) \bigr] \label{eq: 240220202a1} \\
\label{eq: 240220202a2}
& \quad \leq k_2(\rho) \left[
\exp\Bigl( n \rho \, H_{\frac1{1+\rho}}\bigl(f(X)\bigr) \Bigr)
+ \exp\Bigl( n \rho \, H_{\frac1{1+\rho}}\bigl(X \, | \, f(X) \bigr) \Bigr) \right].
\end{align}
\end{subequations}

The bounds in \eqref{eq:raw_bounds_on_moment} are asymptotically tight as $n$
tends to infinity. To see this, note that
\begin{align}
 \lim_{n \to \infty} \frac{1}{n} \ln \left(k_1(\rho)\bigl( 1 + n \ln
 |\set{X}| \bigr)^{-\rho}\right) = 0, \quad
 \lim_{n \to \infty} \frac{1}{n} \ln k_2(\rho) = 0,
\end{align}
and therefore, for all $\rho > 0$,
\begin{align}\label{eq:limit_as_exp_sum}
 &\lim_{n \to \infty} \frac{1}{n} \ln \expectation \bigl[ G_2^{\rho}(X^n) \bigr]\nonumber\\
 &\quad = \lim_{n \to \infty} \frac{1}{n} \log\biggl( \exp\Bigl( n \rho \,
 H_{\frac1{1+\rho}}\bigl(f(X)\bigr) \Bigr) + \exp\Bigl( n \rho \,
 H_{\frac1{1+\rho}}\bigl(X \, | \, f(X) \bigr) \Bigr)\biggr).
\end{align}

Since the sum of the two exponents on the right-hand side (RHS) of \eqref{eq:limit_as_exp_sum}
is dominated by the larger exponential growth rate, it follows that
\begin{align}\label{eq:dominated_by_larger_exponent}
 &\lim_{n \to \infty} \frac{1}{n} \log\left( \exp\Bigl( n \rho \,
 H_{\frac1{1+\rho}}\bigl(f(X)\bigr) \Bigr) + \exp\Bigl( n \rho \,
 H_{\frac1{1+\rho}}\bigl(X \, | \, f(X) \bigr) \Bigr)\right)\nonumber\\
 &\quad = \rho \, \max \Bigl\{ \, H_{\frac1{1+\rho}}\bigl(f(X)\bigr), \;
 H_{\frac1{1+\rho}}\bigl(X \, | \, f(X) \bigr) \Bigr\},
\end{align}
and thus, by \eqref{eq:limit_as_exp_sum} and \eqref{eq:dominated_by_larger_exponent},
\begin{align} \label{eq: 240220202b}
E_2(X; \rho, m, f) &:= \lim_{n \to \infty} \, \tfrac1n \log \expectation[G_2^{\rho}(X^n)] \\
&= \rho \, \max \Bigl\{ \, H_{\frac1{1+\rho}}\bigl(f(X)\bigr), \;
H_{\frac1{1+\rho}}\bigl(X \, | \, f(X) \bigr) \Bigr\}. \label{eq: 11022021a}
\end{align}

As a sanity check, note that if $m = |\set{X}|$ and $f$ is the identity function
$\identity$ (i.e., $\identity(x) = x$ for all $x \in \set{X}$), then $X^n$ is revealed
in Stage 1, and (with Stage 2 obsolete) the $\rho$-th moment of the total number of guesses grows
exponentially with rate
$$\rho \, H_{\frac1{1+\rho}}\bigl(\identity(X)\bigr) = \rho \,
H_{\frac1{1+\rho}}\bigl(X\bigr).$$

This is in agreement with the RHS of \eqref{eq: 240220202b}, as
\begin{align}
 &\rho \, \max \Bigl\{ \, H_{\frac1{1+\rho}}\bigl(\identity(X)\bigr), \;
 H_{\frac1{1+\rho}}\bigl(X \, | \, \identity(X) \bigr) \Bigr\}\nonumber\\
 &\quad = \rho \, \max \Bigl\{ \, H_{\frac1{1+\rho}}\bigl(X\bigr), \;
 H_{\frac1{1+\rho}}\bigl(X \, | \, X \bigr) \Bigr\}\\
 &\quad = \rho \, \max \Bigl\{ \, H_{\frac1{1+\rho}}\bigl(X\bigr), 0 \Bigr\}\\
 &\quad = \rho \, H_{\frac1{1+\rho}}\bigl(X\bigr).
\end{align}

We summarize our results so far in Theorem \ref{theorem1: 2-stage guessing}
below.
\begin{theorem} \label{theorem1: 2-stage guessing}
Let $X^n = (X_1, \ldots, X_n)$ be a sequence of i.i.d. random variables, each drawn according to the PMF $P_X$ of
support $\set{X} := \OneTo{|\set{X}|}$. Let $m \in \FromTo{2}{|\set{X}|-1}$,
$f \in \set{F}_{|\set{X}|,m}$, and define $Y^n := (f(X_1), \ldots, f(X_n))$. When
guessing $X^n$ according to Algorithm 1 (i.e., after first guessing $Y^n$ in descending
order of probability as measured by $P_{Y^n}(\cdot)$ and proceeding
in descending order of probability as measured by $P_{X^n|Y^n}(\cdot|Y^n)$ for guessing $X^n$),
the $\rho$-th moment of the total number of guesses $G_2(X^n)$ satisfies
\begin{enumerate}[a)]
\item
the lower and upper bounds in \eqref{eq:raw_bounds_on_moment} for all $n \in \naturals$
and $\rho > 0$;
\item
the asymptotic characterization \eqref{eq: 240220202b} for
$\rho > 0$ and $n \to \infty$.
\end{enumerate}
\end{theorem}

\subsection*{{A Suboptimal and Simple Construction of $f$ in Algorithm~1 and
Bounds on \boldmath{$\expectation \bigl[ G_2^{\rho}(X^n) \bigr]$}}}

Having established in Theorem \ref{theorem1: 2-stage guessing} that
\begin{equation}\label{eq:approx_total_nr_of_guesses}
 \expectation \bigl[ G_2^{\rho}(X^n) \bigr]
 \approx \exp\bigl(n E_2(X; \rho, m, f)\bigr), \quad \rho > 0,
\end{equation}
we now seek to minimize the exponent $E_2(X; \rho, m, f)$ in the RHS of
\eqref{eq: 11022021a} (for given PMF $P_X$, $\rho > 0$, and
$m \in \FromTo{2}{|\set{X}| - 1}$) over all $f \in \set{F}_{|\set{X}|, m}$.

We proceed by considering a {sub-optimal and simple} construction of $f$, which
enables obtaining explicit bounds as a function of the PMF $P_X$
and the value of $m$, while this construction also does not depend on $\rho$.

For a fixed $m \in \FromTo{2}{|\set{X}|-1}$,
a non-injective deterministic function $f_m^\ast \colon \set{X} \to \OneTo{m}$
is constructed by relying on the Huffman algorithm for lossless compression of $X := X_1$.
This construction also (almost) achieves the maximal mutual information $I(X; f(X))$ among
all deterministic functions $f \colon \set{X} \to \OneTo{m}$ (this issue is
elaborated in the sequel). Heuristically, apart from its simplicity, the motivation
of this sub-optimal construction can be justified since it is expected to reduce the guesswork in Stage~2 of Algorithm~1, where one
wishes to guess $X^n$ on the basis of the knowledge of $Y^n$ with $Y_i = f(X_i)$ for all
$i \in \OneTo{n}$. In this setting, it is shown that the upper and lower bounds
on $\expectation \bigl[ G_2^{\rho}(X^n) \bigr]$ are (almost) asymptotically tight in terms of their exponential growth
rate in $n$. Furthermore, these exponential bounds demonstrate a reduction in the required number
of guesses for $X^n$, as compared to the {optimal one-stage guessing} of $X^n$.

In the sequel, the following construction of a deterministic function
$f_m^\ast \colon \set{X} \to \OneTo{m}$ is analyzed; this construction
was suggested in the proofs of \cite[Lemma~5]{CicaleseGV18} and \cite[Theorem~2]{Sason18b}.

\newpage
{\underline{Algorithm 2 (construction of $f_m^\ast$)}}:
\begin{enumerate}[a)]
\item Let $i := 1, P_1 := P_X$.
\item If $|\supp(P_i)| = m$, let $R := P_i$, and go to Step c. If not, let
$P_{i + 1} := P_i$ with its two least likely symbols merged as in the Huffman code
construction. Let $i \leftarrow i + 1$, and go to Step b.
\item Construct $f_m^\ast \in \set{F}_{|\set{X}|,m}$ by setting $f_m^\ast(k)=j$ if $P_1(k)$ has been merged into $R(j)$.
\end{enumerate}

We now define $Y^{\ast \, n} := (Y_1^\ast, \ldots, Y_n^\ast)$ with
\begin{align} \label{contructed f}
Y_i^\ast := f_m^\ast(X_i), \quad \forall \, i \in \OneTo{n}.
\end{align}

Observe that due to \cite{CicaleseGV18} (Corollary~3 and Lemma~5) and because $f_m^\ast(\cdot)$
operates component-wise on the i.i.d. vector $X^n$, the following lower bound on $\tfrac1n \, I\bigl(X^n; Y^{\ast \, n} \bigr)$ applies:
\begin{align}
\tfrac1n \, I\bigl(X^n; Y^{\ast \, n} \bigr)
&= I\bigl(X; f_m^\ast(X)\bigr) \nonumber \\
&= H\bigl(f_m^\ast(X)\bigr) \nonumber \\
&\geq \max_{f \in \set{F}_{|\set{X}|,m}} H\bigl( f(X) \bigr) - \beta^\ast \nonumber \\
&= \max_{f \in \set{F}_{|\set{X}|,m}} I\bigl(X; f(X)\bigr) - \beta^\ast \nonumber \\
&= \max_{f \in \set{F}_{|\set{X}|,m}} \tfrac1n \, I\bigl(X^n; Y^n \bigr) - \beta^\ast, \label{MI1}
\end{align}
where $Y^n := (f(X_1), \ldots, f(X_n))$ and
\begin{align} \label{beta^ast}
\beta^\ast := \log\biggl(\frac{2}{\mathrm{e} \, \ln 2}\biggr) \approx 0.08607 \, \mbox{bits.}
\end{align}

From the proof of \cite{CicaleseGV18} (Theorem~3), we further have
the following multiplicative~bound:
\begin{align} \label{MI2}
& \tfrac1n \, I\bigl(X^n; Y^{\ast \, n} \bigr) \geq \frac{10}{11}
\, \max_{f \in \set{F}_{|\set{X}|,m}} \frac1n \, I\bigl(X^n; Y^n \bigr).
\end{align}

Note that by \cite{CicaleseGV18} (Lemma~1), the
maximization problem in the RHS of \eqref{MI1} is strongly NP-hard \cite{GareyJ79}.
This means that, unless $\text{P}=\text{NP}$, there is no polynomial-time
algorithm that, given an arbitrarily small $\varepsilon > 0$, produces
a deterministic function $f^{(\varepsilon)} \in \set{F}_{|\set{X}|,m}$~satisfying
\begin{align} \label{strongly NP-hard}
I\bigl(X; f^{(\varepsilon)}(X)\bigr) \geq (1-\varepsilon)
\underset{f \in \set{F}_{|\set{X}|,m}}{\max} I\bigl(X; f(X) \bigr).
\end{align}

We next examine the performance of our candidate function
$f_m^\ast$ when applied in Algorithm~1. To that end, we
first bound $\expectation\bigl[g_{Y^n}^{\rho}(Y^{\ast \, n})\bigr]$
in terms of the R\'{e}nyi entropy of a suitably defined random
variable $\widetilde{X}_m \in \OneTo{m}$ constructed
in Algorithm~3 below. In~the construction, we assume without
loss of generality that $P_X(1) \geq \ldots \geq P_X(|\set{X}|)$
and denote the PMF of $\widetilde{X}_m$ by $Q := R_m(P_X)$.

\vspace*{0.2cm}
{\underline{Algorithm 3 (construction of the PMF $Q := R_m(P_X)$ of the random variable $\widetilde{X}_m$)}}:
\begin{itemize}
\item If $m=1$, then $Q := R_1(P_X)$ is defined to be a point mass at~one;
\item If $m = |\set{X}|$, then $Q := R_{|\set{X}|}(P_X)$ is defined to be equal to $P_X$.
\end{itemize}

Furthermore, for $m \in \FromTo{2}{|\set{X}|-1}$,
\begin{enumerate}[a)]
\item If $P_X(1) < \frac1m$, then $Q$ is defined to be the equiprobable distribution on
$\OneTo{m}$;
\item Otherwise, the PMF $Q$ is defined as
\begin{align} \label{PMF Q}
Q(i) :=
\begin{dcases}
P_X(i), & \mbox{if} \; i \in \OneTo{m^\ast}, \\[0.2cm]
\frac1{m-m^\ast} \sum_{j=m^\ast+1}^{|\set{X}|} P_X(j),
& \mbox{if} \; i \in \{m^\ast+1, \ldots, m\},
\end{dcases}
\end{align}
where $m^\ast$ is the maximal integer $i \in \OneTo{m-1}$,
which satisfies
\begin{align} \label{max integer}
P_X(i) \geq \frac1{m-i} \sum_{j=i+1}^{|\set{X}|} P_X(j).
\end{align}
\end{enumerate}

Algorithm 3 was introduced in \cite{CicaleseGV18,CicaleseGV19}.
The link between the R\'{e}nyi entropy of $\widetilde{X}_m$
and that of $Y^\ast$ was given in Eq.~(34) of \cite{Sason18b}:
\begin{align} \label{range of RE}
H_{\alpha}(Y^\ast) \in \bigl[H_{\alpha}(\widetilde{X}_m)
- v(\alpha), \; H_{\alpha}(\widetilde{X}_m) \bigr],\quad
\forall \alpha > 0,
\end{align}
where
\begin{align} \label{v function}
v(\alpha) :=
\begin{dcases}
\log \left( \frac{\alpha-1}{2^\alpha-2} \right) - \frac{\alpha}{\alpha-1} \,
\log \left(\frac{\alpha}{2^\alpha-1}\right), & \alpha \neq 1 \\[0.1cm]
\log \left(\frac{2}{e \, \ln 2}\right) \approx 0.08607 \; \text{bits}, & \alpha=1.
\end{dcases}
\end{align}

The function $v \colon (0, \infty) \to (0, \log 2)$ is depicted in
Figure~\ref{figure:v_plot} on the following page. It is monotonically
increasing, continuous, and it satisfies
\begin{align}
\lim_{\alpha \downarrow 0} v(\alpha) = 0, \quad
\lim_{\alpha \to \infty} v(\alpha) = \log 2 \; (\mbox{= 1 bit}). \label{limits v}
\end{align}

\begin{figure}[ht]
\vspace*{-4.3cm}
\centerline{\includegraphics[width=11cm]{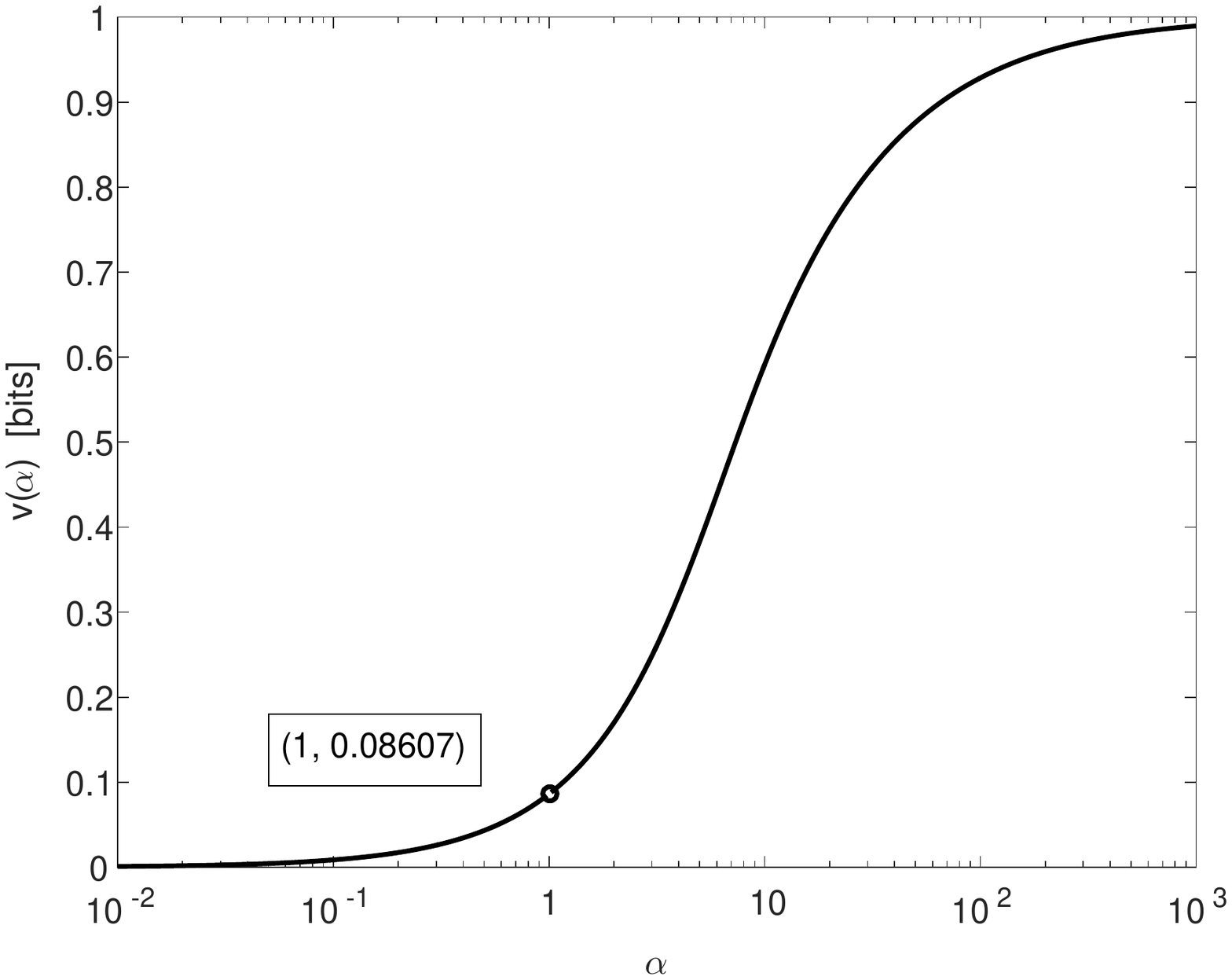}}
\vspace*{-3.6cm}
\caption{\label{figure:v_plot}
A plot of $v \colon (0, \infty) \to (0, \log 2)$, which is monotonically increasing,
continuous, and satisfying~\eqref{limits v}.}
\end{figure}

Combining \eqref{Arikan96-LB1-UB1} and \eqref{range of RE}
yields, for all $\rho > 0$,
\begin{subequations}\label{eq:bounds_Y^n}
\begin{align}
& \bigl( 1 + n \ln m \bigr)^{-\rho} \;
\exp\left( n \left[\rho H_{\frac1{1+\rho}}(\widetilde{X}_m)
- \rho \, v\Bigl(\frac1{1+\rho}\Bigr) \right] \right) \nonumber \\
& \quad \leq \expectation\bigl[g_{Y^n}^{\rho}(Y^{\ast\,n})\bigr] \label{LB-Y^n} \\
\label{UB-Y^n}
& \quad \leq \exp\left( n \rho H_{\frac1{1+\rho}}(\widetilde{X}_m) \right),
\end{align}
\end{subequations}
where, due to \eqref{beta^ast}, the difference between the
exponential growth rates (in $n$) of the lower and upper
bounds in \eqref{eq:bounds_Y^n} is equal to $\rho \, v\bigl(\frac1{1+\rho}\bigr)$,
and it can be verified to satisfy (see \eqref{beta^ast} and~\eqref{v function})
\begin{align} \label{diff}
0 < \rho \, v\Bigl(\frac1{1+\rho}\Bigr) < \beta^\ast \approx 0.08607 \, \log 2, \quad \rho > 0,
\end{align}
where the leftmost and rightmost inequalities of \eqref{diff} are asymptotically tight
as we let $\rho \to 0^{+}$ and $\rho \to \infty$, respectively.

By inserting \eqref{eq:bounds_Y^n} into \eqref{eq:raw_bounds_on_moment} and applying
Inequality \eqref{diff}, it follows that for all $\rho > 0$,
\begin{subequations}\label{eq:bounds-G-X^n}
\begin{align}
& \frac{k_1(\rho)}{\bigl( 1 + n \ln |\set{X}| \bigr)^{\rho}} \, \left[ \,
\exp\left( n \left[\rho H_{\frac1{1+\rho}}(\widetilde{X}_m) -
\beta^\ast \right] \right) +
\exp\Bigl( n \rho \, H_{\frac1{1+\rho}}\bigl(X \, | \, f_m^\ast(X) \bigr) \Bigr)
\right] \nonumber \\[0.1cm]
\label{LB-G-X^n}
& \quad \leq \expectation \bigl[ G_2^{\rho}(X^n) \bigr] \\
\label{UB-G-X^n}
& \quad \leq k_2(\rho) \left[ \exp\left( n \rho H_{\frac1{1+\rho}}(\widetilde{X}_m) \right) +
\exp\Bigl( n \rho \, H_{\frac1{1+\rho}}\bigl(X \, | \, f_m^\ast(X) \bigr) \Bigr) \right].
\end{align}
\end{subequations}

Consequently, by letting $n$ tend to infinity and relying on \eqref{eq: 240220202b},
\begin{subequations}\label{eq:bounds_on_E_2}
\begin{align}
& \max \left\{ \rho H_{\frac1{1+\rho}}(\widetilde{X}_m) -
\beta^\ast, \; \rho \, H_{\frac1{1+\rho}}\bigl(X \, | \, f_m^\ast(X) \bigr) \right\} \nonumber \\
& \quad \leq E_2(X; \rho, m, f_m^\ast) \label{24022020c} \\
& \quad \leq \max \left\{ \rho H_{\frac1{1+\rho}}(\widetilde{X}_m), \;
\rho H_{\frac1{1+\rho}}\bigl(X \, | \, f_m^\ast(X) \bigr) \right\}. \label{24022020d}
\end{align}
\end{subequations}

We next simplify the above bounds by evaluating the maxima in \eqref{eq:bounds_on_E_2}
as a function $m$ and $\rho$. To that end, we use the following lemma.
\begin{lemma} \label{lemma2}
For $\alpha > 0$, let the two sequences $\{a_m(\alpha)\}$
and $\{b_m(\alpha)\}$ be given by
\begin{align}
& a_m(\alpha) := H_{\alpha}\bigl( \widetilde{X}_m \bigr), \label{a sequence} \\
& b_m(\alpha) := H_{\alpha}\bigl( X \, | \, f_m^\ast(X) \bigr), \label{b sequence}
\end{align}
with $m \in \OneTo{|\set{X}|}$. Then,
\begin{enumerate}[a)]
\item The sequence $\{a_m(\alpha)\}$ is monotonically increasing (in $m$), and its first and
last terms are zero and $H_\alpha(X)$, respectively.
\item The sequence $\{b_m(\alpha)\}$ is monotonically decreasing (in $m$), and its first and
last terms are $H_\alpha(X)$ and zero, respectively.
\item If $\supp P_X = \set{X}$, then $\{a_m(\alpha)\}$ is strictly
monotonically increasing, and $\{b_m(\alpha)\}$ is strictly monotonically decreasing.
In particular, for all $m \in \FromTo{2}{|\set{X}|-1}$, $a_m(\alpha)$ and
$b_m(\alpha)$ are positive and strictly smaller than $H_{\alpha}(X)$.
\end{enumerate}
\end{lemma}

\begin{proof}
See Appendix~\ref{appendix: proof of Lemma 2}.
\end{proof}

Since symbols of probability zero (i.e., $x \in \set{X}$ for which $P_X(x) = 0$) do not
contribute to the expected number of guesses, assume without loss of generality that
$\supp P_X = \set{X}$. In view of Lemma~\ref{lemma2}, we can therefore define
\begin{align} \label{M_rho}
m_\rho^\ast &= m_\rho^\ast(P_X)
:= \min \biggl\{ \, m \in \FromTo{2}{|\set{X}|}:
\; a_m\biggl(\frac1{1+\rho}\biggr) \geq
b_m\biggl(\frac1{1+\rho}\biggr) \biggr\}.
\end{align}
Using \eqref{M_rho}, we simplify \eqref{eq:bounds_on_E_2} as follows:
\begin{enumerate}[a)]
\item
If $m < m_\rho^\ast$, then
\begin{align} \label{max1}
E_2(\rho,m) = \rho \, H_{\frac1{1+\rho}}\bigl(X \, | \, f_m^\ast(X) \bigr).
\end{align}
\item
Otherwise, if $m \geq m_\rho^\ast$, then
\begin{align}
\rho H_{\frac1{1+\rho}}(\widetilde{X}_m) - \beta^\ast
\leq E_2(\rho, m) \leq \rho H_{\frac1{1+\rho}}(\widetilde{X}_m). \label{24022020e}
\end{align}
\end{enumerate}

Note that when guessing $X^n$ directly, the $\rho$-th moment of the
number of guesses grows exponentially with rate $\rho H_{\frac1{1+\rho}}(X)$
(cf. \eqref{Arikan96-LB1-UB1}), due to Item~c) in Lemma~\ref{lemma2}, since
\begin{equation}
H_{\frac1{1+\rho}}(\widetilde{X}_m) = a_m\biggl(\frac1{1+\rho}\biggr)
< a_{|\set{X}|}\biggl(\frac1{1+\rho}\biggr) = H_{\frac1{1+\rho}}(X),
\end{equation}
and also because conditioning (on a dependent chance variable) strictly reduces the
R\'{e}nyi entropy \cite{FehrB14}. Eqs.~\eqref{max1} and \eqref{24022020e} imply that for any
$m \in \FromTo{2}{|\set{X}| - 1}$, guessing in two stages according to
Algorithm 1 with $f = f_m^\ast$ reveals $X^n$ sooner (in expectation) than
guessing $X^n$~directly.

\vspace*{0.2cm}
We summarize our findings in this section in the second theorem below.


\begin{theorem} \label{theorem2: 2-stage guessing}
For a given PMF $P_X$ of support $\set{X} = \OneTo{|\set{X}|}$ and $m \in
\FromTo{2}{|\set{X}| - 1}$, let the function $f_m^\ast \in \set{F}_{|\set{X}|, m}$
be constructed according to Algorithm~2, and let the random variable $\widetilde{X}_m$
be constructed according to the PMF in Algorithm~3.
Let $X^n$ be i.i.d. according to $P_X$, and let $Y^n := (f_m^\ast(X_1), \ldots,
f_m^\ast(X_n))$. Finally, for $\rho > 0$, let
\begin{align} \label{E1}
E_1(\rho) := \rho \, H_{\frac1{1+\rho}}(X)
\end{align}
be the optimal exponential growth rate of the $\rho$-th moment single-stage guessing
of $X^n$, and let
\begin{align} \label{E2}
E_2(\rho, m) := E_2(X; \rho, m, f_m^\ast)
\end{align}
be given by \eqref{eq: 11022021a} with $f := f_m^\ast$.
Then, the following holds:
\begin{enumerate}[a)]
\item \cite{CicaleseGV18}:
The maximization of the (normalized) mutual information $\tfrac1n I(X^n; Y^n)$ over
all the deterministic functions $f \colon \set{X} \to \OneTo{m}$ is a strongly NP-hard
problem (for all $n$). However, the deterministic function $f := f_m^\ast$
almost achieves this maximization up to a small additive term, which is equal to
$\beta^\ast := \log\left(\frac{2}{\mathrm{e} \, \ln 2}\right) \approx 0.08607 \, \log 2$,
and also up to a multiplicative term, which is equal to $\tfrac{10}{11}$ (see
\eqref{MI1}--\eqref{MI2}).
\item The $\rho$-th moment of the number of guesses for $X^n$, which is required by the
two-stage guessing in Algorithm~1, satisfies the non-asymptotic bounds
in \eqref{LB-G-X^n}--\eqref{UB-G-X^n}.
\item The asymptotic exponent $E_2(\rho, m)$ satisfies \eqref{max1}--\eqref{24022020e}.
\item For all $\rho > 0$,
\begin{align}
E_1(\rho) - E_2(\rho, m)
\geq \rho \left[ H_{\frac1{1+\rho}}(X) - \max \left\{ H_{\frac1{1+\rho}}(\widetilde{X}_m), \;
H_{\frac1{1+\rho}}\bigl(X \, | \, f_m^\ast(X) \bigr) \right\} \right] > 0,
\end{align}
so there is a reduction in the exponential growth rate (as a function of $n$)
of the required number of guesses for $X^n$ by Algorithm~1 (in comparison to the
optimal one-stage guessing).
\end{enumerate}
\end{theorem}

\section{Two-Stage Guessing: Arbitrary $(X, Y)$}
\label{section: Robert's part}

We next assume that $X^n$ and $Y^n$ are drawn jointly i.i.d. according to a given PMF
$P_{XY}$, and we drop the requirement that Stage~1 need reveal $Y^n$ prior to guessing $X^n$.
Given $\rho > 0$, our goal in this section is to find the least exponential growth rate (in $n$)
of the $\rho$-th moment of the total number of guesses required to recover $X^n$. Since Stage~1
may not reveal $Y^n$, we can no longer express the number of guesses using the ranking functions
$g_{Y^n}(\cdot)$ and $g_{X^n | Y^n}(\cdot | \cdot)$ as in Section~\ref{section: guess Y^n first},
and need new notation to capture the event that $Y^n$ was not guessed in Stage~1.
To that end, let $\set{G}_n$ be a subset of $\set{Y}^n$, and let the ranking function
\begin{equation} \label{eq:tilde_g1}
 \tilde{g}_{Y^n}\colon \set{G}_n \to \OneTo{|\set{G}_n|}
\end{equation}
denote the guessing order in Stage~1 with the understanding that if $Y^n \notin \set{G}_n$, then
\begin{equation} \label{eq:not in G_n}
 \tilde{g}_{Y^n}(Y^n) = |\set{G}_n|
\end{equation}
and the guesser moves on to Stage~2 knowing only that $Y^n \notin \set{G}_n$. We denote the guessing order in Stage~2 by
\begin{equation}\label{eq:tilde_g2}
 \tilde{g}_{X^n | Y^n} \colon \set{X}^n \times \set{Y}^n \to \OneTo{|\set{X}|^n},
\end{equation}
where, for every $y^n \in \set{Y}^n$, $\tilde{g}_{X^n | Y^n}(\cdot | y^n)$ is a ranking function on $\OneTo{|\set{X}|^n}$ that satisfies
\begin{equation}
 \tilde{g}_{X^n | Y^n}(\cdot | y^n) = \tilde{g}_{X^n | Y^n}(\cdot | \eta^n), \quad \forall\, y^n\!, \eta^n \notin \set{G}_n.
\end{equation}

Note that while $ \tilde{g}_{Y^n}$ and $\tilde{g}_{X^n | Y^n}$ depend on $\set{G}_n$, we do not make this dependence explicit.
In the remainder of this section, we prove the following variational characterization of the least exponential growth rate of
the $\rho$-th moment of the total number of guesses in both stages
$\tilde{g}_{Y^n}(Y^n) + \tilde{g}_{X^n | Y^n}(X^n | Y^n)$.
\begin{theorem}\label{thm:main_robert}
 If $(X^n, Y^n)$ are i.i.d. according to $P_{XY}$, then for all $\rho > 0$
 \begin{align}
 & \hspace*{-0.5cm} \lim_{n \to \infty} \min_{\tilde{g}_{Y^n}, \tilde{g}_{X^n | Y^n}}
 \frac{1}{n}\log\expectation\big[\big(\tilde{g}_{Y^n}(Y^n) + \tilde{g}_{X^n | Y^n}(X^n | Y^n)\big)^\rho\big] \nonumber \\
 &= \sup_{Q_{XY}} \Big(\rho\min\big\{H(Q_X), \max\big\{H(Q_Y), H(Q_{X \mid Y})\big\}\big\} - D(Q_{XY} || P_{XY})\Big),\label{eq:target_identity}
 \end{align}
 where the supremum on the RHS of \eqref{eq:target_identity} is over all PMFs $Q_{XY}$ on $\set{X} \times \set{Y}$ (and the limit exists).
\end{theorem}

Note that if $P_{XY}$ is such that $Y = f(X)$, then the RHS of \eqref{eq:target_identity} is less than or equal to the RHS of \eqref{eq:dominated_by_larger_exponent}. In other words, the guessing exponent of Theorem \ref{thm:main_robert} is less than or equal to the guessing exponent of Theorem \ref{theorem1: 2-stage guessing}. This is due to the fact that guessing $Y^n$ in Stage 1 before proceeding to Stage 2 (the strategy examined in Section \ref{section: guess Y^n first}) is just one of the admissible guessing strategies of Section \ref{section: Robert's part} and not necessarily the optimal one.

We prove Theorem \ref{thm:main_robert} in two parts: First, we show that the guesser can be assumed cognizant
of the empirical joint type of $(X^n, Y^n)$; by invoking the law of total expectation, averaging over denominator-$n$ types
$Q_{XY}$ on $\set{X} \times \set{Y}$, we reduce the problem to evaluating the LHS of \eqref{eq:target_identity} under the
assumption that $(X^n, Y^n)$ is drawn uniformly at random from a type class $\tcl{n}(Q_{XY})$ (instead of being i.i.d. $P_{XY}$). We
conclude the proof by solving this reduced problem, showing in particular that when $(X^n, Y^n)$ is drawn uniformly at random from a
type class, the LHS of \eqref{eq:target_identity} can be achieved either by guessing $Y^n$ in Stage~1 or skipping Stage 1 entirely.

We begin with the first part of the proof and show that the guesser can be assumed cognizant of the empirical joint type
of $(X^n, Y^n)$; we formalize and prove this claim in Corollary~\ref{cor:type_given}, which we derive from Lemma~\ref{lmm:cost_of_si}
below.

\begin{lemma}\label{lmm:cost_of_si}
 Let $\tilde{g}_{Y^n}^*$ and $\tilde{g}_{X^n | Y^n}^*$ be ranking functions that minimize the expectation in the LHS of~\eqref{eq:target_identity}, and likewise, let $\tilde{g}_{T; Y^n}^{*}$ and $\tilde{g}_{T; X^n | Y^n}^*$ be ranking functions
 cognizant of the empirical joint type $\Pi_{X^nY^n}$ of $(X^n, Y^n)$ that minimize the expectation in the LHS of
 \eqref{eq:target_identity} over all ranking functions depending on $\Pi_{X^nY^n}$. Then, there exist
 positive constants $a$ and $k$, which are independent of $n$, such~that
 \begin{equation}\label{eq:ub_si}
 \expectation\big[\big(\tilde{g}_{Y^n}^*(Y^n) + \tilde{g}_{X^n | Y^n}^*(X^n | Y^n)\big)^\rho\big]
 \leq \expectation\big[\big(\tilde{g}_{T; Y^n}^*(Y^n) + \tilde{g}_{T; X^n | Y^n}^*(X^n | Y^n)\big)^\rho\big] \, k n^a.
 \end{equation}
\end{lemma}
\begin{proof}
 See Appendix \ref{app:a}.
\end{proof}

\begin{corollary}\label{cor:type_given}
If the limit
 \begin{equation}\label{eq:type_given}
 \lim_{n \to \infty} \frac{1}{n}\log\expectation\big[\big(\tilde{g}_{T; Y^n}^*(Y^n) + \tilde{g}_{T; X^n | Y^n}^*(X^n | Y^n)\big)^\rho\big], \quad \rho > 0,
 \end{equation}
 exists, so does the limit in the LHS of \eqref{eq:target_identity}, and the two are equal.
\end{corollary}
\vspace*{-0.4cm}
\begin{proof}
 By \eqref{eq:ub_si},
 \vspace*{-0.3cm}
 \begin{align}
 &\limsup_{n \to \infty} \min_{\tilde{g}_{Y^n}, \tilde{g}_{X^n | Y^n}}
 \frac{1}{n}\log\expectation\big[\big(\tilde{g}_{Y^n}(Y^n) + \tilde{g}_{X^n | Y^n}(X^n | Y^n)\big)^\rho\big] \nonumber \\
 &\quad = \limsup_{n \to \infty} \frac{1}{n}\log\expectation\big[\big(\tilde{g}_{Y^n}^*(Y^n)
 + \tilde{g}_{X^n | Y^n}^*(X^n | Y^n)\big)^\rho\big]\\
 &\quad \leq \limsup_{n \to \infty} \frac{1}{n}\log\Big(\expectation\big[\big(\tilde{g}_{T; Y^n}^*(Y^n)
 + \tilde{g}_{T; X^n | Y^n}^*(X^n | Y^n)\big)^\rho\big] \, k n^a \Big)\\
 &\quad = \limsup_{n \to \infty} \frac{1}{n}\Big(\log\expectation\big[\big(\tilde{g}_{T; Y^n}^*(Y^n)
 + \tilde{g}_{T; X^n | Y^n}^*(X^n | Y^n)\big)^\rho\big] + \log(k) + a \log(n)\Big)\\
 &\quad = \limsup_{n \to \infty} \frac{1}{n}\log\expectation\big[\big(\tilde{g}_{T; Y^n}^*(Y^n)
 + \tilde{g}_{T; X^n | Y^n}^*(X^n | Y^n)\big)^\rho\big]\\
 &\quad = \lim_{n \to \infty} \frac{1}{n}\log\expectation\big[\big(\tilde{g}_{T; Y^n}^*(Y^n)
 + \tilde{g}_{T; X^n | Y^n}^*(X^n | Y^n)\big)^\rho\big].
 \end{align}

 The inverse inequality
 \begin{align}
 &\liminf_{n \to \infty} \min_{\tilde{g}_{Y^n}, \tilde{g}_{X^n | Y^n}}
 \frac{1}{n}\log\expectation\big[\big(\tilde{g}_{Y^n}(Y^n) + \tilde{g}_{X^n | Y^n}(X^n | Y^n)\big)^\rho\big] \nonumber \\
 &\quad \geq \lim_{n \to \infty} \frac{1}{n}\log\expectation\big[\big(\tilde{g}_{T; Y^n}^*(Y^n)
 + \tilde{g}_{T; X^n | Y^n}^*(X^n | Y^n)\big)^\rho\big]
 \end{align}
 follows from the fact that an optimal guessing strategy depending on the empirical joint type
 $\Pi_{X^nY^n}$ of $(X^n, Y^n)$ cannot be outperformed by a guessing strategy ignorant of $\Pi_{X^nY^n}$.
\end{proof}
Corollary \ref{cor:type_given} states that the minimization in the LHS of \eqref{eq:target_identity}
can be taken over guessing strategies cognizant of the empirical joint type of $(X^n, Y^n)$. As we
show in the next lemma, this implies that evaluating the LHS of \eqref{eq:target_identity} can be
further simplified by taking the expectation with $(X^n, Y^n)$ drawn uniformly at random from a type class
(instead of being i.i.d. $P_{XY}$).

\vspace*{-0.3cm}
\begin{lemma}\label{lmm:expec_over_type}
 Let $\expectation_{Q_{XY}}$ denote expectation with $(X^n, Y^n)$ drawn uniformly at random from the type
 class $\tcl{n}(Q_{XY})$. Then, the following limits exist and
 \begin{align}
 &\lim_{n \to \infty} \min_{\tilde{g}_{T; Y^n}, \tilde{g}_{T; X^n | Y^n}} \frac{1}{n}
 \log\expectation\big[\big(\tilde{g}_{T; Y^n}(Y^n) + \tilde{g}_{T; X^n | Y^n}(X^n | Y^n)\big)^\rho\big]\nonumber\\
 &\quad = \lim_{n \to \infty} \max_{Q_{XY}} \Big(\min_{\tilde{g}_{T; Y^n}, \tilde{g}_{T; X^n | Y^n}}
 \frac{1}{n} \log\expectation_{Q_{XY}}\big[\big(\tilde{g}_{T; Y^n}(Y^n) + \tilde{g}_{T; X^n | Y^n}(X^n | Y^n)\big)^\rho\big] \nonumber \\
 &\quad \hspace*{2.5cm} - D(Q_{XY}\|P_{XY})\Big), \label{eq:limit_over_type}
 \end{align}
 where the maximum in the RHS of \eqref{eq:limit_over_type} is taken over all denominator-$n$ types on
 $\set{X} \times \set{Y}$.
\end{lemma}
\begin{proof}
 Recall that $\Pi_{X^nY^n}$ is the empirical joint type of $(X^n, Y^n)$, and let $\typ{n}(\set{X} \times \set{Y})$
 denote the set of all denominator-$n$ types on $\set{X} \times \set{Y}$. We prove Lemma \ref{lmm:expec_over_type}
 by applying the law of total expectation to the LHS of \eqref{eq:limit_over_type} (averaging over the events
 $\{\Pi_{X^nY^n} = Q_{XY}\}, Q_{XY} \in \typ{n}(\set{X} \times \set{Y})$) and by approximating the probability
 of observing a given type using standard tools from large deviations theory. We first show that the LHS of
 \eqref{eq:limit_over_type} is upper bounded by its RHS:
 \begin{align}
 &\lim_{n \to \infty} \min_{\tilde{g}_{T; Y^n}, \tilde{g}_{T; X^n | Y^n}}
 \frac{1}{n}\log\expectation\big[\big(\tilde{g}_{T; Y^n}(Y^n) + \tilde{g}_{T; X^n | Y^n}(X^n | Y^n)\big)^\rho\big]\nonumber\\[0.1cm]
 &\quad = \lim_{n \to \infty} \frac{1}{n}\log\Bigg(\sum_{Q_{XY}}
 \min_{\tilde{g}_{T; Y^n}, \tilde{g}_{T; X^n | Y^n}} \expectation_{Q_{XY}}\big[\big(\tilde{g}_{T; Y^n}(Y^n)
 + \tilde{g}_{T; X^n | Y^n}(X^n | Y^n)\big)^\rho\big] \nonumber\\
 &\quad \hspace*{2.7cm} \prob[\Pi_{X^nY^n} = Q_{XY}]\Bigg) \\
 &\quad \leq \lim_{n \to \infty} \frac{1}{n}\log\bigg(\max_{Q_{XY}}
 \min_{\tilde{g}_{T; Y^n}, \tilde{g}_{T; X^n | Y^n}} \expectation_{Q_{XY}}\big[\big(\tilde{g}_{T; Y^n}(Y^n)
 + \tilde{g}_{T; X^n | Y^n}(X^n | Y^n)\big)^\rho\big] \nonumber\\
 &\quad \hspace*{2.7cm} \prob[\Pi_{X^nY^n} = Q_{XY}] \, n^\alpha\bigg)\label{eq:expec_over_type_block_1_2}\\
 &\quad = \lim_{n \to \infty} \frac{1}{n}\log\bigg(\max_{Q_{XY}} \min_{\tilde{g}_{T; Y^n},
 \tilde{g}_{T; X^n | Y^n}} \expectation_{Q_{XY}}\big[\big(\tilde{g}_{T; Y^n}(Y^n)
 + \tilde{g}_{T; X^n | Y^n}(X^n | Y^n)\big)^\rho\big] \nonumber\\
 &\quad \hspace*{2.7cm} \prob[\Pi_{X^nY^n} = Q_{XY}]\bigg)\\
 &\quad \leq \lim_{n \to \infty} \frac{1}{n}\log\bigg(\max_{Q_{XY}} \min_{\tilde{g}_{T; Y^n},
 \tilde{g}_{T; X^n | Y^n}} \expectation_{Q_{XY}}\big[\big(\tilde{g}_{T; Y^n}(Y^n)
 + \tilde{g}_{T; X^n | Y^n}(X^n | Y^n)\big)^\rho\big] \nonumber\\
 &\quad \hspace*{2.7cm} 2^{-n D(Q_{XY}\|P_{XY})}\bigg),\label{eq:expec_over_type_block_1_4}
 \end{align}
 where \eqref{eq:expec_over_type_block_1_2} holds for sufficiently large $\alpha$ because the number of types
 grows polynomially in $n$ (see Appendix \ref{app:a}); and \eqref{eq:expec_over_type_block_1_4} follows from
 \cite{Cover_Thomas} (Theorem~11.1.4). We next show that the LHS of \eqref{eq:limit_over_type} is also lower
 bounded by its RHS:
 \begin{align}
 &\lim_{n \to \infty} \min_{\tilde{g}_{T; Y^n}, \tilde{g}_{T; X^n | Y^n}} \frac{1}{n}
 \log\expectation\big[\big(\tilde{g}_{T; Y^n}(Y^n) + \tilde{g}_{T; X^n | Y^n}(X^n | Y^n)\big)^\rho\big]\nonumber\\
 &\quad = \lim_{n \to \infty} \frac{1}{n}\log\Bigg(\sum_{Q_{XY}} \min_{\tilde{g}_{T; Y^n},
 \tilde{g}_{T; X^n | Y^n}} \expectation_{Q_{XY}}\big[\big(\tilde{g}_{T; Y^n}(Y^n)
 + \tilde{g}_{T; X^n | Y^n}(X^n | Y^n)\big)^\rho\big] \nonumber\\
 &\quad \hspace*{2.7cm} \prob[\Pi_{X^nY^n} = Q_{XY}]\Bigg)\\
 &\quad \geq \lim_{n \to \infty} \frac{1}{n}\log\bigg(\max_{Q_{XY}} \min_{\tilde{g}_{T; Y^n},
 \tilde{g}_{T; X^n | Y^n}} \expectation_{Q_{XY}}\big[\big(\tilde{g}_{T; Y^n}(Y^n) +
 \tilde{g}_{T; X^n | Y^n}(X^n | Y^n)\big)^\rho\big] \nonumber\\
 &\quad \hspace*{2.7cm} \prob[\Pi_{X^nY^n} = Q_{XY}] \bigg)\\
 &\quad \geq \lim_{n \to \infty} \frac{1}{n}\log\bigg(\max_{Q_{XY}}
 \min_{\tilde{g}_{T; Y^n}, \tilde{g}_{T; X^n | Y^n}} \expectation_{Q_{XY}}\big[\big(\tilde{g}_{T; Y^n}(Y^n)
 + \tilde{g}_{T; X^n | Y^n}(X^n | Y^n)\big)^\rho\big] \nonumber\\
 &\quad \hspace*{2.7cm} 2^{-n [D(Q_{XY}\|P_{XY})+\delta_n]} \bigg)\label{eq:expec_over_type_block_2_3}\\
 &\quad \geq \lim_{n \to \infty} \frac{1}{n}\log\bigg(\max_{Q_{XY}} \min_{\tilde{g}_{T; Y^n},
 \tilde{g}_{T; X^n | Y^n}} \expectation_{Q_{XY}}\big[\big(\tilde{g}_{T; Y^n}(Y^n) +
 \tilde{g}_{T; X^n | Y^n}(X^n | Y^n)\big)^\rho\big]\nonumber\\
 &\quad \hspace*{2.7cm} 2^{-n D(Q_{XY}\|P_{XY})}\bigg),
 \end{align}
 where in \eqref{eq:expec_over_type_block_2_3},
 \begin{align} \label{delta_n}
 \delta_n := \frac{(|\set{X}||\set{Y}|-1) \log(n+1)}{n}
 \end{align}
 tends to zero as we let $n$ tend to infinity, and the inequality in \eqref{eq:expec_over_type_block_2_3}
 follows again from \cite{Cover_Thomas} (Theorem~11.1.4). Together, \eqref{eq:expec_over_type_block_1_4} and
 \eqref{eq:expec_over_type_block_2_3} imply the equality in \eqref{eq:limit_over_type}.
\end{proof}
We now have established the first part of the proof of Theorem~\ref{thm:main_robert}: In Corollary~\ref{cor:type_given},
we showed that the ranking functions $\tilde{g}_{Y^n}$ and $\tilde{g}_{X^n | Y^n}$ can be assumed cognizant of the
empirical joint type of $(X^n, Y^n)$, and in Lemma~\ref{lmm:expec_over_type}, we showed that under this assumption,
the minimization of the $\rho$-th moment of the total number of guesses can be carried out with $(X^n, Y^n)$ drawn uniformly
at random from a type class.

Below, we give the second part of the proof: We show that if the pair $(X^n, Y^n)$ is drawn uniformly
at random from $\tcl{n}(Q_{XY})$, then

\vspace*{-1cm}
\begin{align}\label{eq:optimal_exponent_tcl}
 &\lim_{n \to \infty} \min_{\tilde{g}_{Y^n}, \tilde{g}_{X^n | Y^n}}
 \frac{1}{n}\log\expectation\big[\big(\tilde{g}_{Y^n}(Y^n) + \tilde{g}_{X^n | Y^n}(X^n | Y^n)\big)^\rho\big] \nonumber \\
 & \quad = \rho \, \min\big\{H(Q_X), \max\{H(Q_Y), H(Q_{X | Y})\big\} \big\}, \quad \rho > 0.
\end{align}

Note that Corollary~\ref{cor:type_given} and Lemma~\ref{lmm:expec_over_type} in conjunction with \eqref{eq:optimal_exponent_tcl}
conclude the proof of Theorem~\ref{thm:main_robert}:

\vspace*{-1cm}
\begin{align}
 &\lim_{n \to \infty} \min_{\tilde{g}_{Y^n}, \tilde{g}_{X^n | Y^n}} \frac{1}{n}
 \log\expectation\big[\big(\tilde{g}_{Y^n}(Y^n) + \tilde{g}_{X^n | Y^n}(X^n | Y^n)\big)^\rho\big]\nonumber\\[-0.1cm]
 &\quad = \lim_{n \to \infty} \min_{\tilde{g}_{T; Y^n}, \tilde{g}_{T; X^n | Y^n}}
 \frac{1}{n}\log\expectation\big[\big(\tilde{g}_{T; Y^n}(Y^n) + \tilde{g}_{T; X^n | Y^n}(X^n | Y^n)\big)^\rho\big]\\[-0.1cm]
 &\quad = \lim_{n \to \infty} \max_{Q_{XY}} \Big(\min_{\tilde{g}_{T; Y^n}, \tilde{g}_{T; X^n | Y^n}}
 \frac{1}{n} \log\expectation_{Q_{XY}}\big[\big(\tilde{g}_{T; Y^n}(Y^n) + \tilde{g}_{T; X^n | Y^n}(X^n | Y^n)\big)^\rho\big] \nonumber \\
 &\quad \hspace*{2.5cm} - D(Q_{XY}\|P_{XY})\Big)\label{eq:pro_final}\\
 &\quad = \sup_{Q_{XY}} \Big(\rho\min\big\{H(Q_X), \max\big\{H(Q_Y), H(Q_{X \mid Y})\big\}\big\} - D(Q_{XY} || P_{XY})\Big), \label{2407a}
\end{align}
where the supremum in the RHS of \eqref{2407a} is taken over all PMFs $Q_{XY}$ on $\set{X} \times \set{Y}$, and the step follows from \eqref{eq:pro_final} because the set of types is dense in the set of all PMFs (in the same sense that $\mathbb{Q}$ is dense in $\mathbb{R}$).

It thus remains to prove \eqref{eq:optimal_exponent_tcl}. We begin with the direct part:
Note that when $(X^n, Y^n)$ is drawn uniformly at random from $\tcl{n}(Q_{XY})$, the $\rho$-th moment
of the total number of guesses grows exponentially with rate $\rho H(Q_X)$ if we skip Stage~1
and guess $X^n$ directly and with rate $\rho\max\big\{H(Q_Y), H(Q_{X | Y})\big\}$ if we guess
$Y^n$ in Stage~1 before moving on to guessing $X^n$. To prove the second claim, we argue by case
distinction on $\rho$. Assuming first that $\rho \leq 1$,
\begin{align}
 & \hspace*{-0.5cm} \lim_{n \to \infty} \frac{1}{n}\log\expectation\big[\big(\tilde{g}_{T; Y^n}(Y^n)
 + \tilde{g}_{T; X^n | Y^n}(X^n | Y^n)\big)^\rho\big] \nonumber \\[0.1cm]
 &\leq \lim_{n \to \infty} \frac{1}{n}\log\expectation\big[\tilde{g}_{T; Y^n}(Y^n)^\rho
 + \tilde{g}_{T; X^n | Y^n}(X^n | Y^n)^\rho\big]\label{sequential_strategy_block_1_1}\\[0.1cm]
 &= \lim_{n \to \infty} \frac{1}{n}\log\Big(\!\expectation\big[\tilde{g}_{T; Y^n}(Y^n)^\rho\big]
 + \expectation\big[\tilde{g}_{T; X^n | Y^n}(X^n | Y^n)^\rho\big]\Big)\\[0.1cm]
 &\leq \lim_{n \to \infty} \frac{1}{n}\log\Big(2^{n\rho H(Q_Y)} + 2^{n\rho H(Q_{X | Y})}\Big)
 \label{sequential_strategy_block_1_2}\\
 &= \rho\max\big\{H(Q_Y), H(Q_{X | Y})\big\} \label{sequential_strategy_block_1_3},
\end{align}
where \eqref{sequential_strategy_block_1_1} holds because $\rho \leq 1$ (see Lemma~\ref{lemma1}); \eqref{sequential_strategy_block_1_2}
holds since, by the latter assumption, $Y^n$ is revealed at the end of Stage~1, and thus, the guesser (cognizant of $Q_{XY}$)
will only guess elements from the conditional type class $\tcl{n}(Q_{X | Y} | Y^n)$; and \eqref{sequential_strategy_block_1_3}
follows from the fact that the exponential growth rate of a sum of two exponents is dominated by the larger one.
We wrap up the argument by showing that the LHS of \eqref{sequential_strategy_block_1_1} is also lower bounded by the
RHS of \eqref{sequential_strategy_block_1_3}:
\begin{align}
 &\hspace*{-0.5cm} \lim_{n \to \infty} \frac{1}{n}\log\expectation\big[\big(\tilde{g}_{T; Y^n}(Y^n)
 + \tilde{g}_{T; X^n | Y^n}(X^n | Y^n)\big)^\rho\big] \nonumber \\
 &= \lim_{n \to \infty} \frac{1}{n}\log\expectation\bigg[2^\rho\Big(\tfrac12 \, \tilde{g}_{T; Y^n}(Y^n)
 + \tfrac12 \, \tilde{g}_{T; X^n | Y^n}(X^n | Y^n) \Big)^\rho\bigg]\label{sequential_strategy_block_2_1}\\
 &\geq \lim_{n \to \infty} \frac{1}{n}\log\expectation\big[2^{\rho - 1}\big(\tilde{g}_{T; Y^n}(Y^n)^\rho
 + \tilde{g}_{T; X^n | Y^n}(X^n | Y^n)^\rho\big)\big]\label{sequential_strategy_block_2_2}\\[0.1cm]
 &= \lim_{n \to \infty} \frac{1}{n}\log\Big(\!\expectation\big[\tilde{g}_{T; Y^n}(Y^n)^\rho\big]
 + \expectation\big[\tilde{g}_{T; X^n | Y^n}(X^n | Y^n)^\rho\big]\Big)\\
 &\geq \lim_{n \to \infty} \frac{1}{n}\log\bigg(\frac{1}{1 + \rho}\Big(2^{n\rho H(Q_Y)}
 + 2^{n\rho H(Q_{X | Y})}\Big)2^{-n \rho \delta_n}\bigg)\label{sequential_strategy_block_2_3}\\
 &= \rho\max\big\{H(Q_Y), H(Q_{X | Y})\big\},
\end{align}
where \eqref{sequential_strategy_block_2_2} follows from Jensen's inequality; and \eqref{sequential_strategy_block_2_3}
follows from \cite{moser19_12} (Proposition~6.6) and the lower bound on the size of a (conditional)
type class \cite{Cover_Thomas} (Theorem~11.1.3).
The case ${\rho > 1}$ can be proven analogously and is hence omitted (with the term $2^{\rho-1}$ in the RHS of
\eqref{sequential_strategy_block_2_2} replaced by~one). Note that by applying the better of the two proposed
guessing strategies (i.e., depending on $Q_{XY}$, either guess $Y^n$ in Stage~1 or skip it) the $\rho$-th
moment of the total number of guesses grows exponentially with rate
$\rho\min\big\{H(Q_X), \max\big\{H(Q_Y), H(Q_{X | Y}\big)\big\} \big\}$. This concludes the direct part of
the proof of \eqref{eq:optimal_exponent_tcl}. We remind the reader that while we have constructed a guessing
strategy under the assumption that the empirical joint type $\Pi_{X^nY^n}$ of $(X^n, Y^n)$ is known,
Lemma~\ref{lmm:cost_of_si} implies the existence of a guessing strategy of equal asymptotic performance that
does not depend on $\Pi_{X^nY^n}$. Moreover, Lemma \ref{lmm:cost_of_si} is constructive in the sense that the
type-independent guessing strategy can be explicitly derived from the type-cognizant one (cf. the proof of
Proposition~6.6 in \cite{moser19_12}).

We next establish the converse of \eqref{eq:optimal_exponent_tcl} by showing that when $(X^n, Y^n)$
is drawn uniformly at random from $\tcl{n}(Q_{XY})$,
\begin{align}
 &\hspace*{-0.5cm} \liminf_{n \to \infty} \frac{1}{n}\log\expectation\big[\big(\tilde{g}_{T; Y^n}(Y^n)
 + \tilde{g}_{T; X^n | Y^n}(X^n | Y^n)\big)^\rho\big] \nonumber \\
 &\geq \rho\min\big\{H(Q_X), \max\big\{H(Q_Y), H(Q_{X | Y})\big\}\big\} \label{eq:optimal_exponent}
\end{align}
for all two-stage guessing strategies. To see why \eqref{eq:optimal_exponent} holds, consider an
arbitrary guessing strategy, and let the sequence $n_1, n_2, \ldots$ be such that
\begin{align}
 &\hspace*{-0.5cm} \lim_{k \to \infty} \frac{1}{n_k}\log\expectation\big[\big(\tilde{g}_{T; Y^{n_k}}(Y^{n_k})
 + \tilde{g}_{T; X^{n_k} | Y^{n_k}}(X^{n_k} | Y^{n_k})\big)^\rho\big] \nonumber \\
 &= \liminf_{n \to \infty} \frac{1}{n}\log\expectation\big[\big(\tilde{g}_{T; Y^n}(Y^n)
 + \tilde{g}_{T; X^n | Y^n}(X^n | Y^n)\big)^\rho\big] \label{eq:liminf_sequence}
\end{align}
and such that the limit
\begin{equation}\label{eq:def_alpha}
 \alpha := \lim_{k \to \infty} \frac{|\set{G}_{n_k}|}{|\tcl{n_k}(Q_Y)|}
\end{equation}
exists. Using Lemma \ref{lmm:insufficient_guesses} below, we show that the LHS of \eqref{eq:liminf_sequence}
(and thus, also the LHS of \eqref{eq:optimal_exponent}) is lower bounded by $\rho H(Q_X)$ if $\alpha = 0$ and by
$\rho\max\big\{H(Q_Y), H(Q_{X | Y})\big\}$ if $\alpha > 0$. This establishes the converse, because the lower of
the two bounds must apply in any case.

\begin{lemma}\label{lmm:insufficient_guesses}
If the pair $(X^n, Y^n)$ is drawn uniformly at random from a type class $\tcl{n}(Q_{XY})$ and
 \begin{equation}\label{eq:insufficient_first_stage}
 \limsup_{n \to \infty} \frac{|\set{G}_n|}{|\tcl{n}(Q_Y)|} = 0,
 \end{equation}
 then
 \begin{equation}\label{eq:full_second_stage_exp}
 \lim_{n \to \infty} \frac{1}{n}\log\expectation\big[\tilde{g}_{X^n | Y^n}(X^n | Y^n)^\rho\big] = \rho H(Q_X),
 \end{equation}
 where $Q_X$ and $Q_Y$ denote the $X$- and $Y$-marginal of $Q_{XY}$.
\end{lemma}
\begin{proof}
 Note that since
 \begin{equation}
 \lim_{n \to \infty} \frac{1}{n}\log\big|\tcl{n}(Q_X)\big| = H(Q_X),
 \end{equation}
 the RHS of \eqref{eq:full_second_stage_exp} is trivially upper bounded by $\rho H(Q_X)$. It thus suffices to show that
 \eqref{eq:insufficient_first_stage} yields the lower bound
 \begin{equation}\label{eq:full_second_stage_exp_forced}
 \liminf_{n \to \infty} \frac{1}{n}\log\expectation\big[\tilde{g}_{X^n | Y^n}(X^n | Y^n)^\rho\big] \geq \rho H(Q_X).
 \end{equation}
 To show that \eqref{eq:insufficient_first_stage} yields \eqref{eq:full_second_stage_exp_forced}, we define the indicator variable
 \begin{equation}
 E_n := \begin{dcases}
 0, &\text{if } Y^n \in \set{G}_n\\
 1, &\text{else},
 \end{dcases}
 \end{equation}
 and observe that due to \eqref{eq:insufficient_first_stage} and the fact that $Y^n$ is drawn uniformly at random from $\tcl{n}(Q_Y)$,
 \begin{equation}\label{eq:prb_to_zero}
 \lim_{n \to \infty} \prob[E_n = 1] = 1.
 \end{equation}
 Consequently, $H(E_n)$ tends to zero as $n$ tends to infinity, and because
 \begin{equation}
 H(X^n) - H(X^n \mid E_n) = I(X^n; E_n) \leq H(E_n),
 \end{equation}
we get
\begin{equation}
\lim_{n \to \infty} \big(H(X^n) - H(X^n \mid E_n)\big) = 0.
\end{equation}
 This and \eqref{eq:prb_to_zero} imply that
\begin{equation}\label{eq:same_limit}
\lim_{n \to \infty} \frac{1}{n} H(X^n) = \lim_{n \to \infty} \frac{1}{n} H(X^n \mid E_n = 1).
\end{equation}
To conclude the proof of Lemma \ref{lmm:insufficient_guesses}, we proceed as follows:
\newpage
 \begin{align}
 &\liminf_{n \to \infty} \frac{1}{n}\log\expectation[\tilde{g}_{X^n | Y^n}(X^n | Y^n)^\rho] \nonumber \\
 &\quad \geq \liminf_{n \to \infty} \frac{1}{n}\log\expectation[\tilde{g}_{X^n | Y^n}(X^n | Y^n)^\rho | E_n = 1]\label{eq:insufficient_guesses_block_1_1}\\
 &\quad \geq \liminf_{n \to \infty} \frac{1}{n}\, \rho\Rentr{1/(1 + \rho)}(X^n \mid E_n = 1)\label{eq:insufficient_guesses_block_1_2}\\
 &\quad \geq \liminf_{n \to \infty} \frac{1}{n}\, \rho H(X^n \mid E_n = 1)\label{eq:insufficient_guesses_block_1_3}\\
 &\quad = \liminf_{n \to \infty} \frac{1}{n}\, \rho H(X^n)\label{eq:insufficient_guesses_block_1_4}\\
 &\quad = \rho H(Q_X),
 \end{align}
 where \eqref{eq:insufficient_guesses_block_1_1} holds due to \eqref{eq:prb_to_zero} and the law of total expectation;
 \eqref{eq:insufficient_guesses_block_1_2} follows from \cite{Arikan96} (Theorem~1); \eqref{eq:insufficient_guesses_block_1_3}
 holds because the R\'enyi entropy is monotonically decreasing in its order and because $\rho > 0$; and \eqref{eq:insufficient_guesses_block_1_4} is due to \eqref{eq:same_limit}.
\end{proof}
We now conclude the proof of the converse part of \eqref{eq:optimal_exponent_tcl}. Assume first that $\alpha$ (as defined in \eqref{eq:def_alpha}) equals zero. By \eqref{eq:liminf_sequence} and Lemma \ref{lmm:insufficient_guesses},
\begin{align}
 &\liminf_{n \to \infty} \frac{1}{n}\log\expectation\big[\big(\tilde{g}_{T; Y^n}(Y^n) + \tilde{g}_{T; X^n | Y^n}(X^n | Y^n)\big)^\rho\big]\nonumber\\[0.1cm]
 &\quad = \lim_{k \to \infty} \frac{1}{n_k}\log\expectation\big[\big(\tilde{g}_{T; Y^{n_k}}(Y^{n_k}) + \tilde{g}_{T; X^{n_k} | Y^{n_k}}(X^{n_k} | Y^{n_k})\big)^\rho\big]\\[0.1cm]
 &\quad \geq \liminf_{k \to \infty} \frac{1}{n_k}\log\expectation\big[\tilde{g}_{T; Y^{n_k}}(X^{n_k} | Y^{n_k})^\rho\big]\\[0.1cm]
 &\quad = \rho H(Q_X),
\end{align}
establishing the first contribution to the RHS of \eqref{eq:optimal_exponent}. Next, let $\alpha > 0$. Applying \cite{moser19_12} (Proposition 6.6) in conjunction with \eqref{eq:def_alpha} and the fact that $Y^{n_k}$ is drawn uniformly at random from $\tcl{n_k}(Q_Y)$,
\begin{equation}\label{eq:eventually_tcl}
 \expectation\big[\tilde{g}_{T; Y^{n_k}}(Y^{n_k})^\rho\big] \geq \frac{\alpha}{2} \cdot \frac{2^{{n_k}\rho H(Q_Y)}2^{-n_k\delta_{n_k}}}{1 + \rho}
\end{equation}
for all sufficiently large $k$. Using \eqref{eq:eventually_tcl} and proceeding analogously as in \eqref{sequential_strategy_block_2_1} to \eqref{sequential_strategy_block_2_3}, we now establish the second contribution to the RHS of \eqref{eq:optimal_exponent}:
\begin{align}
 &\liminf_{n \to \infty} \frac{1}{n}\log\expectation\big[\big(\tilde{g}_{T; Y^n}(Y^n) + \tilde{g}_{T; X^n | Y^n}(X^n | Y^n)\big)^\rho\big]\nonumber\\[0.1cm]
 &\quad = \lim_{k \to \infty} \frac{1}{n_k}\log\expectation\big[\big(\tilde{g}_{T; Y^{n_k}}(Y^{n_k}) + \tilde{g}_{T; X^{n_k} | Y^{n_k}}(X^{n_k} | Y^{n_k})\big)^\rho\big]\\[0.2cm]
 &\quad \geq \liminf_{k \to \infty} \frac{1}{n_k}\log\Big(\!\expectation\big[\tilde{g}_{T; Y^{n_k}}(Y^{n_k})^\rho\big] + \expectation\big[\tilde{g}_{T; X^{n_k} | Y^{n_k}}(X^{n_k} | Y^{n_k})^\rho\big]\Big)\\[0.2cm]
 &\quad \geq \liminf_{k \to \infty} \frac{1}{n_k}\log\Bigg(\!\frac{\alpha}{2} \cdot \frac{2^{{n_k}\rho H(Q_Y)}2^{-n_k\delta_{n_k}}}{1 + \rho} + \expectation\big[\tilde{g}_{T; X^{n_k} | Y^{n_k}}(X^{n_k} | Y^{n_k})^\rho\big]\Bigg)\label{eq:one_of_two_strategies_block_1_3}\\[0.1cm]
 &\quad \geq \liminf_{k \to \infty} \frac{1}{n_k}\log\Bigg(\!\frac{\alpha}{2} \cdot \frac{2^{{n_k}\rho H(Q_Y)}2^{-n_k\delta_{n_k}}}{1 + \rho} + \frac{2^{n_k\rho H(Q_{X | Y})}2^{-n_k\delta_{n_k}}}{1 + \rho}\Bigg)\label{eq:one_of_two_strategies_block_1_4}\\[0.1cm]
 &\quad = \rho\max\big\{H(Q_Y), H(Q_{X | Y})\big\},
\end{align}
where \eqref{eq:one_of_two_strategies_block_1_3} is due to \eqref{eq:eventually_tcl}; and in \eqref{eq:one_of_two_strategies_block_1_4},
we granted the guesser access to $Y^n$ at the beginning of Stage~2.

\section{Summary and Outlook}\label{sec:conclusion}
We proposed a new variation on the Massey--Arikan guessing problem where, instead of guessing $X^n$ directly, the guesser is allowed to first produce guesses of a correlated ancillary sequence $Y^n$. We characterized the least achievable exponential growth rate (in $n$) of the $\rho$-th moment of the total number of guesses in the two stages when $X^n$ is i.i.d. according to $P_X$, $Y_i = f(X_i)$ for all $i \in [1:n]$, and the guesser must recover $Y^n$ in Stage 1 before proceeding to Stage 2 (Section \ref{section: guess Y^n first}, Theorems \ref{theorem1: 2-stage guessing} and \ref{theorem2: 2-stage guessing}); and when the pair $(X^n, Y^n)$ is jointly i.i.d. according to $P_{XY}$ and Stage 1 need not reveal $Y^n$ (Section \ref{section: Robert's part}, Theorem \ref{thm:main_robert}). Future directions of this work include:

\begin{enumerate}[1)]
 \item The generalization of our results to a larger class of sources (e.g., Markov sources);
 \item A study of the information-like properties of the guessing exponents \eqref{eq:first} and \eqref{eq:second};
 \item Finding the optimal block-wise description $Y^n = f(X^n)$ and its associated two-stage guessing exponent;
 \item The generalization of the cryptographic problems \cite{ArikanM99,BracherHL_IT19} to a setting where the adversary may also produce guesses of leaked side information.
\end{enumerate}

\section*{Acknowledgment}
The authors are indebted to Amos Lapidoth for his contribution
to Section~\ref{section: Robert's part} (see~\cite{GrackzykL19}).
The constructive comments in the review process, which helped to improve the
presentation, are gratefully acknowledged.

\newpage
\appendices
\renewcommand{\thesubsection}{\Alph{section}.\arabic{subsection}}

\section{Proof of Lemma~\ref{lemma1}}
\label{appendix: proof of Lemma 1}

\setcounter{equation}{0}
\setcounter{corollary}{0}
\renewcommand{\theequation}{\thesection\arabic{equation}}
\renewcommand{\thecorollary}{\thesection\arabic{corollary}}

If $\rho \geq 1$, then
\begin{align}
\left( \sum_{i=1}^k a_i \right)^\rho &= k^\rho \left( \frac1k
\sum_{i=1}^k a_i \right)^\rho \nonumber \\
&\leq k^\rho \cdot \frac1k \sum_{i=1}^k a_i^\rho \label{lemma1: 2} \\
&= k^{\rho-1} \sum_{i=1}^k a_i^\rho, \nonumber
\end{align}
where \eqref{lemma1: 2} holds by Jensen's inequality and since
the mapping $x \mapsto x^\rho$ for $x \geq 0$ is convex.
If at least one of the non-negative
$a_i$'s is positive (if all $a_i$'s are zero, it is trivial), then
\begin{align}
\left( \sum_{i=1}^k a_i \right)^\rho &= \left( \sum_{j=1}^k a_j \right)^\rho
\, \sum_{i=1}^k \left( \frac{a_i}{\sum_{j=1}^k a_j} \right) \nonumber \\
&\geq \left( \sum_{j=1}^k a_j \right)^\rho \, \sum_{i=1}^k
\left( \frac{a_i}{\sum_{j=1}^k a_j} \right)^\rho \label{lemma1: 3} \\
&= \sum_{i=1}^k a_i^\rho,
\end{align}
where \eqref{lemma1: 3} holds since $0 \leq \frac{a_i}{\sum_{j=1}^k a_j} \leq 1$
for all $i \in \OneTo{n}$, and $\rho \geq 1$.
If $\rho \in (0,1)$, then inequalities \eqref{lemma1: 2} and
\eqref{lemma1: 3} are reversed. The conditions for equalities
in \eqref{lemma1: 1} are easily verified.

\section{Proof of Lemma~\ref{lemma2}}
\label{appendix: proof of Lemma 2}

\setcounter{equation}{0}
\setcounter{corollary}{0}
\renewcommand{\theequation}{\thesection\arabic{equation}}
\renewcommand{\thecorollary}{\thesection\arabic{corollary}}

From \eqref{PMF Q}, $R_1(P_X)$ is a unit probability mass at~one and
$R_{|\set{X}|}(P_X) \equiv P_X$, so \eqref{a sequence} gives~that
\begin{align}
a_1(\alpha) = 0, \quad a_{|\set{X}|}(\alpha) = H_\alpha(X).
\end{align}
In view of Lemma~5 in \cite{Sason18b}, it follows that for all
$m \in \OneTo{|\set{X}|-1}$,
\begin{align}
\label{1:38 in Sason18b}
a_{m+1}(\alpha) &= \max_{Q \in \set{P}_{m+1} \,
: \, P_X \prec Q} H_{\alpha}(Q) \\
\label{inclusion1}
&\geq \max_{Q \in \set{P}_m \, : \, P_X \prec Q} H_{\alpha}(Q) \\
\label{2:38 in Sason18b}
&= a_m(\alpha),
\end{align}
where \eqref{1:38 in Sason18b} and \eqref{2:38 in Sason18b} are due
to \cite{Sason18b} (38) and \eqref{a sequence}; \eqref{inclusion1}
holds since $\set{P}_m \subset \set{P}_{m+1}$.

We next prove Item~b. Consider the sequence of functions
$\{f_m^\ast\}_{m=1}^{|\set{X}|}$, defined over the set
$\set{X}$. By construction (see Algorithm~2), $f_{|\set{X}|}^\ast$
is the identity function since all the respective $|\set{X}|$ nodes
in the Huffman algorithm stay un-changed in this case.
We also have $f_1^\ast(x) = 1$ for all $x \in \set{X}$
(in the latter case, by Algorithm~2, all nodes are merged by the
Huffman algorithm into a single node). Hence, from \eqref{b sequence},
\begin{align}
b_1(\alpha) = H_{\alpha}(X), \quad b_{|\set{X}|}(\alpha) = 0.
\end{align}
Consider the construction of the function $f_m^\ast$ by Algorithm~2.
Since the transition from $m+1$ to $m$ nodes is obtained by merging
two nodes without affecting the other $m-1$ nodes, it follows by the
data processing theorem for the Arimoto--R\'{e}nyi conditional entropy
(see \cite{FehrB14} (Theorem~2 and Corollary~1)) that, for all
$m \in \OneTo{|\set{X}|-1}$,
\begin{align} \label{by DPI}
b_{m+1}(\alpha) = H_{\alpha}\bigl( X \, | \, f_{m+1}^\ast(X) \bigr)
\leq H_{\alpha}\bigl( X \, | \, f_m^\ast(X) \bigr) = b_m(\alpha).
\end{align}

We finally prove Item~c. Suppose that $P_X$ is supported on
the set $\set{X}$. Under this assumption, it follows from
the strict Schur concavity of the R\'{e}nyi entropy that
the inequality in \eqref{inclusion1} is strict, and therefore,
\eqref{1:38 in Sason18b}--\eqref{2:38 in Sason18b} imply that
$a_m(\alpha) < a_{m+1}(\alpha)$ for all $m \in \OneTo{|\set{X}|-1}$.
In particular, Item~a implies that $0 < a_m(\alpha) < H_{\alpha}(X)$
holds for every $m \in \FromTo{2}{|\set{X}|-1}$.
Furthermore, the conditioning on $f_{m+1}^\ast(X)$ enables
distinguishing between the two labels of $\set{X}$, which correspond to
the pair of nodes that are being merged (by the Huffman algorithm)
in the transition from $f_{m+1}^\ast(X)$ to $f_m^\ast(X)$. Hence,
the inequality in \eqref{by DPI} turns out to be strict under the
assumption that $P_X$ is supported on the set $\set{X}$. In
particular, under that assumption, it follows from Item~b that
$0 < b_m(\alpha) < H_{\alpha}(X)$ holds for every
$m \in \FromTo{2}{|\set{X}|-1}$.

\section{Proof of Lemma \ref{lmm:cost_of_si}}\label{app:a}
\setcounter{equation}{0}
\setcounter{corollary}{0}
\renewcommand{\theequation}{\thesection\arabic{equation}}
\renewcommand{\thecorollary}{\thesection\arabic{corollary}}

We prove Lemma~\ref{lmm:cost_of_si} as a consequence of Corollary~\ref{corr:general_cost_of_si} below
and the fact that the number of denominator-$n$ types on a finite set grows polynomially in $n$
(\cite{Cover_Thomas}, Theorem~11.1.1).
\begin{corollary} (Moser \cite{moser19_12}, (6.47) and Corollary~6.10)\label{corr:general_cost_of_si}
 Let the random triple $(U, V, W)$ take values in the finite set $\set{U} \times \set{V} \times \set{W}$,
 and let $\big(\tilde{g}^*_U(\cdot), \tilde{g}^*_{U | V}(\cdot | \cdot)\big)$ and
 $\big(\tilde{g}^*_{U | W}(\cdot | \cdot), \tilde{g}^*_{U | V, W}(\cdot | \cdot, \cdot)\big)$
 be ranking functions that, for a given $\rho > 0$, minimize
 \begin{equation}
 \expectation\big[\big(\tilde{g}_{U}(U) + \tilde{g}_{U | V}(U | V)\big)^\rho\big]
 \end{equation}
 over all two-stage guessing strategies (with no access to $W$) and
 \begin{equation}
 \expectation\big[\big(\tilde{g}_{U | W}(U | W) + \tilde{g}_{U | V, W}(U | V, W)\big)^\rho\big]
 \end{equation}
 over all two-stage guessing strategies cognizant of $W$. Then,
 \begin{equation}
 \expectation\big[\big(\tilde{g}^*_{U}(U) + \tilde{g}^*_{U | V}(U | V)\big)^\rho\big] \leq \expectation\big[\big(\tilde{g}^*_{U | W}(U | W)
 + \tilde{g}^*_{U | V, W}(U | V, W)\big)^\rho\big] \, |\set{W}|^\rho.
 \end{equation}
\end{corollary}
Lemma \ref{lmm:cost_of_si} follows from Corollary \ref{corr:general_cost_of_si} with $U \leftarrow X^n,\, \set{U} \leftarrow \set{X}^n$; $V \leftarrow Y^n,\, \set{V} \leftarrow \set{Y}^n$; $W \leftarrow \Pi_{X^nY^n}$; $\set{W} \leftarrow \typ{n}(\set{X} \times \set{Y})$,
and by noticing that for all $n \in \naturals$,
\begin{equation}
 |\typ{n}(\set{X} \times \set{Y})|^\rho \leq (n + 1)^{\rho (|\set{X} \times \set{Y}| -1)} \leq k \, n^a,
\end{equation}
where $a := \rho \; \bigl(|\set{X} \times \set{Y}| -1 \bigr)$ and $k := 2^a$ are positive constants independent of $n$.

\end{document}